\documentclass[final]{dmtcs-episciences}

\usepackage{geometry}
\usepackage{multicol}

\usepackage{amsmath}
\usepackage{amsfonts}
\usepackage{amssymb}
\usepackage{amsthm}
\usepackage{verbatim}

\usepackage{graphicx}
\usepackage{epic}
\usepackage{eepic}

\usepackage{epsfig,float}
\usepackage{pdfsync}

\usepackage{gastex}
\usepackage{multicol}

\usepackage{xcolor}

\renewcommand{\le}{\leqslant}
\renewcommand{\ge}{\geqslant}

\newcommand{\ol}{\overline}
\newcommand{\eps}{\varepsilon}
\newcommand{\emp}{\emptyset}

\newcommand{\Sig}{\Sigma}
\newcommand{\sig}{\sigma}
\newcommand{\noin}{\noindent}

\newcommand{\bi}{\begin{itemize}}
\newcommand{\ei}{\end{itemize}}
\newcommand{\be}{\begin{enumerate}}
\newcommand{\ee}{\end{enumerate}}
\newcommand{\bd}{\begin{description}}
\newcommand{\ed}{\end{description}}
\newcommand{\bq}{\begin{quote}}
\newcommand{\eq}{\end{quote}}

\newcommand{\cA}{{\mathcal A}}

\newcommand{\cD}{{\mathcal D}}

\newcommand{\cN}{{\mathcal N}}
\newcommand{\cP}{{\mathcal P}}

\newcommand{\cT}{{\mathcal T}}

\newcommand{\one}{{\mathbf 1}}

\newcommand{\lraL}{{\hspace{.1cm}{\leftrightarrow_L} \hspace{.1cm}}}

\newtheorem{theorem}{Theorem}

\theoremstyle{definition}
\newtheorem{definition}{Definition}
\newtheorem{example}{Example}

\theoremstyle{remark}
\newtheorem{remark}{Remark}

\DeclareMathOperator{\lcm}{lcm} 

\usepackage{color}

\title{Most Complex Regular Ideal Languages}

\author{Janusz~Brzozowski\affiliationmark{1}\thanks{This work was supported by the Natural Sciences and Engineering Research Council of Canada under grant No.~OGP0000871.}
\and Sylvie Davies\affiliationmark{2} \and Bo Yang Victor Liu\affiliationmark{1}} 

\affiliation{David R. Cheriton School of Computer Science, University of Waterloo \\
Department of Pure Mathematics, University of Waterloo}


\received{2015-12-03}
\accepted{2016-10-11}
\revised{2016-10-04}

\begin{document}
\publicationdetails{18}{2016}{3}{15}{1343}

\maketitle
\begin{abstract}
A right ideal (left ideal, two-sided ideal) is a non-empty language $L$ over an alphabet $\Sigma$ 
such that
$L=L\Sigma^*$ ($L=\Sigma^*L$, $L=\Sigma^*L\Sigma^*$).
Let $k=3$ for right ideals, 4 for left ideals and 5 for two-sided ideals.
We show that there exist sequences ($L_n \mid n \ge k $) of  right, left, and two-sided regular ideals, where $L_n$ has quotient complexity (state complexity) $n$, such that 
$L_n$ is most complex in its class under the following measures of complexity: 
the size of the syntactic semigroup,
the quotient complexities of the left quotients of $L_n$,  
the number of atoms (intersections of complemented and uncomplemented left quotients), 
the quotient complexities of the atoms,  
and the quotient complexities  of 
reversal, 
star,  
product (concatenation), 
and all binary boolean operations.
In that sense, these ideals are ``most complex'' languages in their classes, or ``universal witnesses'' to the complexity of the various operations.
\end{abstract}

\section{Introduction}
\label{sec:introduction}

We begin informally, postponing definitions until Section~\ref{sec:background}.
In~\cite{Brz13} Brzozowski introduced a list of conditions that a regular language should satisfy in order to be called ``most  complex'', and found a sequence $(U_n \mid n \ge 3)$ of  regular languages with quotient/state complexity $n$ that have the smallest possible alphabet and meet all of these conditions~\cite{Brz13}.
Namely, the languages $U_n$ meet the upper bounds for 
the size of the syntactic semigroup,
the quotient complexities of left quotients,
the number of atoms (intersections of complemented and uncomplemented left quotients),
the quotient complexities of the atoms, 
and the quotient complexities of  the following operations:
reversal, 
star, 
product (concatenation), 
and all binary boolean operations.
In this sense the languages in this sequence are most complex when compared to other regular languages of the same quotient complexity.
However, these  ``universal witnesses''
cannot be used when studying the complexities listed above in subclasses of regular languages, since they generally lack the properties of those classes.

This paper is part of an ongoing project to investigate whether the approach used for general regular languages can be extended to subclasses. 
We present sequences of most complex languages for the classes of right, left, and two-sided regular  ideals. 
Right ideals were chosen as an initial ``test case'' for this project due to their simple structure; we were able to obtain a sequence of most complex right ideals by making small modifications to the sequence $(U_n \mid n \ge 3)$.
A preliminary version of our results about right ideals appeared in~\cite{BrDa14}. 
The sequences of witnesses for left and two-sided ideals are more complicated, and first appeared in~\cite{BrYe11} where they were conjectured to have syntactic semigroups of maximal size; this was later proved in~\cite{BSY15}.
Our main new result is a demonstration that these sequences are in fact most complex.
It has been shown in~\cite{BrSz15a} that a most complex sequence does not exist for the class of suffix-free languages.

Having a single sequence of witnesses for all the complexity measures is useful when one needs to test systems that perform operations on regular languages and finite automata: to determine the sizes of the largest automata that can be handled by the system, one can use the same sequence of witnesses for all the operations.

For a further discussion of regular ideals see~\cite{BrDa14,BJL13,BrSz14,BSY15,BrYe11}.
It was pointed out in~\cite{BJL13} that 
ideals deserve to be studied for several reasons. 
They are fundamental objects in semigroup theory. They appear in the theoretical computer science literature 
as early as 1965,
and continue to be of interest; an overview of historical and recent work on ideal languages is given in~\cite{BJL13}.
Besides being of theoretical interest, ideals also play a role in algorithms for pattern matching: 
For example, when searching for all words ending in a word from some set $L$, one is looking for all the words of the left ideal $\Sig^*L$.
Additional examples of the use of ideals in applications can be found in~\cite{AhCo75,CrHa90,CHL07,YCDLK06}. 

\section{Background}
\label{sec:background}

A \emph{deterministic finite automaton (DFA)} $\cD= (Q,\Sig,\delta,q_1,F)$ consists of 
a finite non-empty set $Q$ of \emph{states},
a finite non-empty \emph{alphabet} $\Sig$,  
a \emph{transition function} $\delta\colon Q\times \Sig\to Q$, an
\emph{initial state} $q_1\in Q$, and 
a set $F\subseteq Q$ of \emph{final states}.
The transition function is extended to functions $\delta'\colon Q \times \Sig^*\to Q$ and $\delta''\colon 2^Q \times \Sig^*\to 2^Q$ as usual, but these extensions are  also denoted by $\delta$. 
State $q \in Q$  is \emph{reachable} if there is a word $w\in\Sig^*$ such that $\delta(q_1,w)=q$. 
The \emph{language  accepted} by $\cD$ is $L(\cD)=\{w\in\Sig^* \mid \delta(q_1,w)\in F\}$.
Two DFAs are \emph{equivalent} if their languages are equal.
The \emph{language of a state} $q$ is the language accepted by the DFA 
$(Q,\Sig,\delta,q,F)$.
Two states are \emph{equivalent} if their languages are equal; otherwise, they are \emph{distinguishable} by some word that is in the language of one of the states, but not the other. 
A~DFA is \emph{minimal} if all of its states are reachable and no two states are equivalent.
A state is \emph{empty} if its language is empty.

A \emph{nondeterministic finite automaton (NFA)} is a tuple $\cN= (Q,\Sig,\eta,Q_I,F)$, where 
$Q$, $\Sig$, and $F$ are as in a DFA, $\eta\colon Q\times \Sig\to 2^Q$ is the transition function 
and $Q_I\subseteq Q$ is the \emph{set of initial states}.
An \emph{$\eps$-NFA} has all the features of an NFA but its transition function 
$\eta\colon Q\times (\Sig\cup \{\eps\})\to 2^Q$ allows also transitions under the empty word $\eps$. 
The \emph{language accepted} by an NFA or an $\eps$-NFA is the set of words $w$ for which there exists a sequence of transitions such that the concatenation of the symbols causing the transitions is $w$,
and this sequence leads from a state in $Q_I$ to a state in $F$.
Two NFAs are \emph{equivalent} if they accept the same language.

Without loss of generality we use the set $Q_n=\{1,2,\dots,n\}$ as the set of states for our automata.
A \emph{transformation} of $Q_n$ is a mapping of $Q_n$ into itself.
We denote the image of a state $q$ under a transformation $t$ by $qt$.
An arbitrary transformation of $Q_n$ can be written as
\goodbreak

\begin{equation*}\label{eq:transmatrix}
t=\left( \begin{array}{ccccc}
1 & 2 &   \cdots &  n-1 & n \\
p_1 & p_2 &   \cdots &  p_{n-1} & p_n
\end{array} \right ),
\end{equation*}
where $p_q = qt$,  $1\le q\le n$, and $p_q\in Q_n$.
The \emph{image of a set} $P\subseteq Q_n$ 
is $Pt=\{ pt \mid p\in P \}$.

The \emph{identity} transformation $\one$ maps  each element to itself, that is, $q\one=q$ for $q=1,\ldots,n$.
For $k\ge 2$,  a
transformation $t$ of a set $P=\{p_1,\dots,p_k\}$ is a \emph{$k$-cycle} if there exist pairwise different elements $p_1,\ldots,p_k$ such that
$p_1t=p_2, p_2t=p_3,\ldots, p_{k-1}t=p_k$,  $p_kt=p_1$, and all other elements of $Q_n$ are mapped to themselves.
A $k$-cycle is denoted by $(p_1,p_2,\ldots,p_k)$.
A \emph{transposition} is a $2$-cycle $(p,q)$.
A transformation that changes only one element $p$ to an element $q\neq p$ is denoted by $(p\to q)$.
A transformation mapping a subset $P$ of $Q_n$ to a single element $q$ and acting as the identity on $Q_n\setminus P$ is denoted by $(P \rightarrow q)$.

A \emph{permutation} of $Q_n$ is a mapping of $Q_n$ \emph{onto} itself. 
The set of all permutations of a set $Q_n$ of $n$ elements is a group, 
called the \emph{symmetric group} of degree $n$. 
This group has size $n!$.
It is well known that  two generators are necessary and sufficient to generate the symmetric group of degree $n$;
 in particular, the pairs $\{(1,2,\ldots, n),(1,2)\}$ and $\{(1,2,\ldots, n),(2,3,\dots,n)\}$ generate $S_n$.

The set $\cT_{Q_n}$ of all transformations of  $Q_n$
is a semigroup, in fact a monoid with $\one$ as the identity. 
It is well known that three transformations are necessary and sufficient to generate  $\cT_{Q_n}$;
in particular,  the triples
 $\{(1,2,\ldots, n),(1,2),(n \rightarrow 1) \}$ and $\{(1,2,\ldots, n),(2,3,\dots,n),(n \rightarrow 1)\}$ generate $\cT_{Q_n}$.

Let $\cD = (Q, \Sig, \delta, q_0, F)$ be a DFA. For each word $w \in \Sig^*$, the transition function induces a transformation $\delta_w$ of $Q$ by  $w$: for all $q \in Q$, 
$q\delta_w = \delta(q, w).$ 
The set $T_{\cD}$ of all such transformations by non-empty words forms a semigroup of transformations called the \emph{transition semigroup} of $\cD$~\cite{Pin97}. 
Conversely, we can use a set  $\{\delta_a \mid a \in \Sig\}$ of transformations to define $\delta$, and so the DFA $\cD$. We write $a\colon t$, where $t$ is a transformation of $Q$, to mean that the transformation $\delta_a$ induced by $a$ is~$t$. 

The \emph{Myhill congruence}~\cite{Myh57} $\lraL$ of a language $L\subseteq \Sig^*$ is defined on $\Sig^+$ as follows:
$$
\mbox{For $x, y \in \Sig^+$, } x \lraL y \mbox{ if and only if } wxz\in L  \Leftrightarrow wyz\in L\mbox { for all } w,z \in\Sig^*.
$$
This congruence is also known as the \emph{syntactic congruence} of $L$.
The quotient set $\Sig^+/ \lraL$ of equivalence classes of the relation $\lraL$ is a semigroup called the \emph{syntactic semigroup} of $L$.
If  $\cD$ is a minimal DFA of $L$, then $T_{\cD}$ is isomorphic to the syntactic semigroup $T_L$ of $L$~\cite{Pin97}, and we represent elements of $T_L$ by transformations in~$T_{\cD}$. 
The size of the syntactic semigroup has been used as a measure of complexity for regular languages~\cite{Brz13,BrYe11,HoKo04}, and is denoted by $\sig(L)$. 

The \emph{Nerode right congruence}~\cite{Ner58} of a language $L\subseteq \Sig^*$ is defined on $\Sig^*$ as follows:
$$
\mbox{For $x, y \in \Sig^*$, } x \rightarrow_L  y \mbox{ if and only if } xz\in L  \Leftrightarrow yz\in L\mbox { for all } z \in\Sig^*.
$$
The (left) quotient of $L\subseteq \Sig^*$ by a word $w\in \Sig^*$ is the language 
$w^{-1}L=\{x\in\Sig^*\mid wx\in L\}$.
Thus two words $x$ and $y$ are in the same class of the Nerode right congruence if they define the same quotient, that is, if $x^{-1}L=y^{-1}L$, and  
the number of equivalence classes of $\rightarrow_L$ is the number of quotients, which is called the 
 \emph{quotient complexity}~\cite{Brz10} of $L$.
An equivalent concept is
the \emph{state complexity of a regular language}~\cite{Yu01} $L$, which  is the number of states in a minimal DFA with alphabet $\Sig$ that recognizes $L$. 
This paper uses  the term \emph{complexity} for both of these equivalent notions.
We denote the (quotient/state) complexity of a regular language $L$ by $\kappa(L)$.

Atoms of regular languages were studied in~\cite{BrTa14}, and their complexities  
in~\cite{BrDa14,BrTa13,Iva16}.
Consider the left congruence defined as follows:
$$
\mbox{ For $x, y \in \Sig^+$, } x \leftarrow_L y  \mbox{ if and only if }  wx\in L  \Leftrightarrow wy\in L
\mbox { for all } w\in\Sig^*.
$$
Thus $x \leftarrow_L y$ if $x\in w^{-1}L \text{ if and only if } y\in w^{-1}L$. 
An equivalence class of this relation is called an \emph{atom}.
It follows that atoms are intersections of complemented and uncomplemented quotients.
In particular, if the quotients of $L$ are $K_1,\dotsc,K_n$, then for each atom $A$ there is a unique set $S \subseteq \{K_1,\dotsc,K_n\}$ such that
$A = \bigcap_{K \in S} K \cap \bigcap_{K \not\in S} (\Sig^* \setminus K)$.
In~\cite{Brz13} it was argued that 
for a regular language to be considered ``most complex'' when compared with other languages of the same (quotient/state) complexity, it should have the maximal possible number of atoms and each atom should have maximal complexity.
The complexity of atoms of ideals was studied in~\cite{BrDa14,BrDa15}, and we shall  state the results obtained there without proofs.

Most of the results in the literature concentrate on the (quotient/state) complexity of operations on regular languages.
The \emph{complexity of an operation} is the maximal  complexity of the language resulting from the operation as a function of the  complexities of the arguments.

It is generally assumed when studying complexity of binary operations that both arguments are languages over the same alphabet, since if they have different alphabets, we can just view them as languages over the union of the alphabets.
However, in 2016, Brzozowski demonstrated that this viewpoint leads to incorrect complexity results for operations on languages over different alphabets~\cite{Brz16}. He introduced a notion of \emph{unrestricted (quotient/state) complexity} of operations, which gives correct results when languages have different alphabets. The traditional notion of complexity, in which languages are assumed to have the same alphabet, is referred to as \emph{restricted (quotient/state) complexity} of operations.
As this paper was written well before~\cite{Brz16}, all our results are in terms of restricted complexity.
A study of unrestricted complexity of binary operations on ideals can be found in~\cite{BrCo16}.

There are two parts to the process of  establishing the  complexity of an operation.
First, one must find an \emph{upper bound} on the  complexity of the result of the operation
by using quotient computations or automaton constructions.
Second, one must find \emph{witnesses}  that meet this upper bound.
One usually defines a sequence $(L_n\mid n\ge k)$ of languages, where  $k$ is some small positive integer (the bound may not apply for small values of $n$).  This sequence is called a \emph{stream}. The languages  in a stream usually differ only in the parameter $n$. 
For example, 
one might study unary languages $(\{a^n\}^*\mid n\ge 1)$ that have zero $a$'s modulo~$n$. 
A unary operation then takes its argument from a stream $(L_n\mid n\ge k)$.
For a binary operation, one adds as the second argument a stream $(L'_m\mid m\ge k)$. 

Sometimes one considers the case where the inputs to the operations are restricted to some \emph{subclass} of the class of regular languages. In this setting, typically the upper bounds on complexity are different and different witnesses must be found. The complexity of operations on regular ideal languages was studied in~\cite{BJL13}.

While witness streams are normally different for different operations, the main result of this paper shows that for the subclasses of right, left and two-sided ideals, the complexity bounds for all ``basic operations'' (those mentioned in the introduction) can be met by a single stream of languages along with a stream of ``dialects'', which are slightly modified versions of the languages. Several types of dialects were introduced in~\cite{Brz13}; in this paper we consider only dialects defined as follows:

Let $\Sig=\{a_1,\dots,a_k\}$ be an alphabet; we assume that its elements are ordered as shown.
Let $\pi$ be a \emph{partial permutation} of $\Sig$, that is, a partial function $\pi \colon \Sig \rightarrow \Gamma$ where $\Gamma \subseteq \Sig$, for which there exists $\Delta \subseteq \Sig$ such that $\pi$ is bijective when restricted to $\Delta$ and  undefined on $\Sig \setminus \Delta$. We denote undefined values of $\pi$ by the symbol ``$-$''.

If $L$ is a language over $\Sig$, we denote it by $L(a_1,\dots,a_k)$ to stress its dependence on $\Sig$.
If $\pi$ is a partial permutation, let $s_\pi$ be the language substitution  defined as follows: 
 for $a\in \Sig$, 
$a \mapsto \{\pi(a)\}$ when $\pi(a)$ is defined, and $a \mapsto \emp$ when $\pi(a)$ is not defined.
For example, if $\Sig=\{a,b,c\}$, $L(a,b,c)=\{a,b,c\}^*\{ab, acc\}$, and $\pi(a)=c$, $\pi(b)=-$, and $\pi(c)=b$, then $s_\pi(L)= \{b,c\}^*\{cbb\}$.
In other words, the letter $c$ plays the role of $a$, and $b$ plays the role of $c$.
A \emph{permutational dialect} of $L(a_1,\dots,a_k)$ is a language of the form 
$s_\pi(L(a_1,\dots,a_k))$, where $\pi$ is a partial permutation of $\Sig$; this dialect is denoted by
$L(\pi(a_1),\dotsc,\pi(a_k))$.
If the order on $\Sig$ is understood, we use  $L(\Sig)$ for $L(a_1,\dots,a_k)$ and $L(\pi(\Sig))$
for $L(\pi(a_1),\dotsc,\pi(a_k))$.

Let $\Sig=\{a_1,\dots,a_k\}$,  and 
let $\cD = (Q,\Sig,\delta,q_1,F)$ be a DFA; we denote it by
$\cD(a_1,\dots,a_k)$ to stress its dependence on $\Sig$.
If $\pi$ is a partial permutation, then the \emph{permutational dialect} 
$$\cD(\pi(a_1),\dotsc,\pi(a_k))$$ of
$\cD(a_1,\dots,a_k)$ is obtained by changing the alphabet of $\cD$ from $\Sig$ to $\pi(\Sig)$, and modifying $\delta$ so that in the modified DFA 
$\pi(a_i)$ induces the transformation induced by $a_i$  in the original DFA; thus $\pi(a_i)$ plays the role of $a_i$.
One verifies that if the language $L(a_1,\dots,a_k)$ is accepted by DFA $\cD(a_1,\dots,a_k)$, then
$L(\pi(a_1),\dotsc,\pi(a_k))$ is accepted by $\cD(\pi(a_1),\dotsc,\pi(a_k))$.

In the sequel we refer to permutational dialects simply as dialects.

\begin{example}
Suppose $\cD = \cD(a,b,c) = (\{1,2,3\},\{a,b,c\},\delta,q_1,F)$, where  $\delta$ is defined by the transformations 
$a \colon (1,2,3)$, $b \colon (3 \rightarrow 1)$, and $c \colon (1,2)$. 
 Let $L = L(a,b,c)$ be the language of this DFA. 
Consider the partial permutation $\pi(a) = b$, $\pi(b) = -$, and $\pi(c) = a$. In the dialect associated with $\pi$, the letter $b$ plays the role of $a$, and $a$ plays the role of $c$. 
Thus $\cD(b,-,a)$
 is the DFA $(\{1,2,3\},\{a,b\},\delta^\pi,q_1,F)$, where $\delta^\pi$ is defined by 
$a \colon (1,2)$ and $b \colon (1,2,3)$.
The language of $\cD(b,-,a)$ is the dialect $L(b,-,a)$ of $L$.
\end{example}

The notion of a \emph{most complex stream} of regular languages was introduced informally in~\cite{Brz13}. 
A~most complex stream is one whose languages together with  their dialects  meet all the upper bounds for the complexity measures described in the introduction.
We now make this notion precise. First, however, we recall a property of boolean functions.
Let the truth values of propositions be 1 (true) and 0 (false). Let $\circ\colon \{0,1\}\times \{0,1\} \to \{0,1\}$ be a binary boolean function.
Extend $\circ$ to a function
$\circ\colon 2^{\Sig^*}\times 2^{\Sig^*}\to 2^{\Sig^*}$:
If $w\in \Sig^*$ and $L,L'\subseteq \Sig^*$, 
then $w\in (L\circ L') \Leftrightarrow (w\in L) \circ (w\in L').$
A binary boolean function is \emph{proper} if it is not a constant  and not a function of one variable only.

\begin{definition}
Let $C$ be a class of languages and let $C_n$ be the subclass of $C$ that consists of all the languages of $C$ that have (quotient/state) complexity $n$.
Let $\Sig=\{a_1,\dots,a_k\}$, and let $(L_n(\Sig) \mid n\ge k)$ be a stream of languages, where $L_n\in C_n$ for all $n\ge k$. Then $(L_n(\Sig) \mid n\ge k)$ is most complex in class $C$ if it satisfies all of the following conditions:
\be
\item
The syntactic semigroup of $L_n(\Sig)$ has maximal cardinality for each $n\ge k$.
\item
Each quotient of $L_n(\Sig)$ has maximal   complexity for each $n\ge k$.
\item
$L_n(\Sig)$ has the maximal possible number of atoms for each $n\ge k$.
\item
Each atom of $L_n(\Sig)$ has maximal complexity for each $n\ge k$.
\item
The reverse of $L_n(\Sig)$ has maximal complexity for each $n\ge k$.
\item 
The star of $L_n(\Sig)$ has maximal  complexity for each $n\ge k$.
\item
The product
$L_m(\Sig)L_n(\Sig)$ has maximal  complexity for all $m,n\ge k$.
\item
There exists a dialect $L_n(\pi(\Sig))$ such that each proper binary boolean function 
$L_m(\Sig)\circ L_n(\pi(\Sig))$ has maximal  complexity for all $m,n\ge k$.
\ee
\end{definition}

A most complex stream $(L_n \mid n \ge 3)$ for the class of regular languages was introduced in~\cite{Brz13}. 
We give the definition of $L_n$ below.

\begin{definition}
\label{def:regular}
For $n\ge 3$, let $\cD_n=\cD_n(a,b,c)=(Q_n,\Sig,\delta_n, 1, \{n\})$, where 
$\Sig=\{a,b,c\}$, 
and $\delta_n$ is defined by the transformations
$a\colon (1,\dots,n)$,
$b\colon(1,2)$,
${c\colon(n \rightarrow 1)}$. 
Let $L_n=L_n(a,b,c)$ be the language accepted by~$\cD_n$.
The structure of $\cD_n(a,b,c)$ is shown in Figure~\ref{fig:RegWit}. 
\end{definition}

\begin{figure}[ht]
\unitlength 11pt
\begin{center}\begin{picture}(31,7)(-2,0)
\gasset{Nh=2.4,Nw=2.4,Nmr=1.2,ELdist=0.3,loopdiam=1.2}

\node(1)(2,4){$1$}\imark(1)
\node(2)(6,4){$2$}
\node(3)(10,4){$3$}
\node[Nframe=n](qdots)(14,4){$\dots$}
\node(n-2)(18,4){{\small $n-2$}}
\node(n-1)(22,4){{\small $n-1$}}
\node(n)(26,4){{\small $n$}}\rmark(n)

\drawedge(1,2){$a,b$}
\drawedge(2,3){$a$}
\drawedge(3,qdots){$a$}
\drawedge(qdots,n-2){$a$}
\drawedge(n-2,n-1){$a$}
\drawedge(n-1,n){$a$}
\drawedge[curvedepth=-1.9,ELdist=-.9](2,1){$b$}
\drawedge[curvedepth=3.8](n,1){$a,c$}

\drawloop(1){$c$}
\drawloop(2){$c$}
\drawloop(3){$b,c$}
\drawloop(n-1){$b,c$}
\drawloop(n-2){$b,c$}
\drawloop(n){$b$}

\end{picture}\end{center}
\caption{Minimal DFA $\cD_n(a,b,c)$  of Definition~\ref{def:regular}.}
\label{fig:RegWit}
\end{figure}

Our main contributions in this paper are most complex streams for the classes of right, left, and two-sided regular ideals.

\section{Right Ideals}

A stream of right ideals  that is most complex was defined and studied in~\cite{BrDa14}. For completeness we include the results from that paper; the proof of Theorem~\ref{thm:RBool2} did not appear in~\cite{BrDa14}.

\begin{definition}
\label{def:RWit}
For $n\ge 3$, let $\cD_n=\cD_n(a,b,c,d)=(Q_n,\Sig,\delta_n, 1, \{n\})$, where 
$\Sig=\{a,b,c,d\}$, 
and $\delta_n$ is defined by the transformations
$a\colon (1,\dots,n-1)$,
$b\colon(2,\ldots,n-1)$,
${c\colon(n-1\rightarrow 1)}$
and ${d\colon(n-1\rightarrow n)}$. 
Note that $b$ is the identity when $n=3$.
Let $L_n=L_n(a,b,c,d)$ be the language accepted by~$\cD_n$.
The structure of $\cD_n(a,b,c,d)$ is shown in Figure~\ref{fig:RWit}. 
\end{definition}


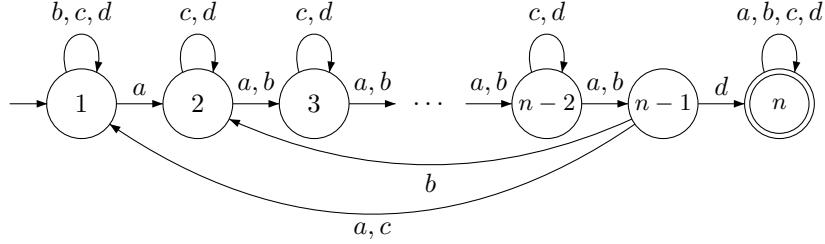
\begin{figure}[ht]
\unitlength 11pt
\begin{center}\begin{picture}(31,7)(-2,0)
\gasset{Nh=2.4,Nw=2.4,Nmr=1.2,ELdist=0.3,loopdiam=1.2}

\node(1)(2,4){$1$}\imark(1)
\node(2)(6,4){$2$}
\node(3)(10,4){$3$}
\node[Nframe=n](qdots)(14,4){$\dots$}
\node(n-2)(18,4){{\small $n-2$}}
\node(n-1)(22,4){{\small $n-1$}}
\node(n)(26,4){{\small $n$}}\rmark(n)

\drawedge(1,2){$a$}
\drawedge(2,3){$a,b$}
\drawedge(3,qdots){$a,b$}
\drawedge(qdots,n-2){$a,b$}
\drawedge(n-2,n-1){$a,b$}
\drawedge(n-1,n){$d$}
\drawedge[curvedepth=2.1](n-1,2){$b$}
\drawedge[curvedepth=3.8](n-1,1){$a,c$}

\drawloop(1){$b,c,d$}
\drawloop(2){$c,d$}
\drawloop(3){$c,d$}
\drawloop(n-2){$c,d$}
\drawloop(n){$a,b,c,d$}

\end{picture}\end{center}
\caption{Minimal DFA $\cD_n(a,b,c,d)$  of Definition~\ref{def:RWit}.}
\label{fig:RWit}
\end{figure}

The DFA of Figure~\ref{fig:RWit} has a similar structure to the DFA of Figure~\ref{fig:RegWit}.
More precisely,
DFA  $\cD_n$ of Figure~\ref{fig:RWit} is constructed by taking DFA $\cD_{n-1}$ of Figure~\ref{fig:RegWit}, adding a new state $n$ and a new input $d\colon  (n-1~\rightarrow~n)$, making $n$ the only final state, and having $b$ induce the cyclic permutation $(2,\dotsc,n-1)$, rather than  the transposition $(1,2)$. The new state and input $d$ ensure that $L_n$ is a right ideal. The new transformation by $b$ is necessary since, if $b$  induces $(1,2)$ in $\cD_n$, then $L_n$ does not meet the bound for product.

\begin{theorem}[Right Ideals~\cite{BrDa14}]
\label{thm:main}
For each $n\ge 3$, the DFA $\cD_n(a,b,c,d)$ of Definition~\ref{def:RWit} is minimal and its 
language $L_n(a,b,c,d)$ is a right ideal of complexity $n$.
The stream $(L_n(a,b,c,d) \mid n \ge 3)$  with dialect stream
$(L_n(b,a,c,d) \mid n \ge 3)$
is most complex in the class of regular right ideals.
In particular, this stream meets all the complexity bounds listed below, which are maximal for right ideals.
In several cases the bounds can be met with restricted alphabets, as shown below.
\begin{enumerate}
\item
The syntactic semigroup of $L_n(a,b,c,d)$ has cardinality $n^{n-1}$.  Moreover, fewer than four inputs do not suffice to meet this bound.
\item
The quotients of $L_n(a,-,-,d)$ have complexity $n$, except for the quotient $\{a,d\}^*$, which has complexity 1.
\item
$L_n(a,-,-,d)$ has $2^{n-1}$ atoms.
\item
For each atom $A_S$ of $L_n(a,b,c,d)$, the complexity $\kappa(A_S)$ satisfies:
\begin{equation*}
	\kappa(A_S) =
	\begin{cases}
		2^{n-1}, 			& \text{if $S=Q_n$;}\\
		1 + \sum_{x=1}^{|S|}\sum_{y=1}^{n-|S|}\binom{n-1}{x-1}\binom{n-x}{y-1},
		 			& \text{if $\emp \subsetneq S \subsetneq Q_n$.}
		\end{cases}
\end{equation*}
\item
The reverse of $L_n(a,-,-,d)$ has complexity $2^{n-1}$.
\item
The star of $L_n(a,-,-,d)$ has complexity $n+1$.
\item
The product $L_m(a,b,-,d) L_n(a,b,-,d)$ has complexity $m+2^{n-2}$.
\item
For any proper binary boolean function $\circ$, the complexity of $L_m(a,b,-,d) \circ L_n(b,a,-,d)$
is maximal. In particular,
	\be
	\item
	$L_m(a,b,-,d) \cap L_n(b,a,-,d)$ and $L_m(a,b,-,d) \oplus 	L_n(b,a,-,d)$ have complexity  $mn$.
	\item
	$L_m(a,b,-,d) \setminus L_n(b,a,-,d)$ has complexity $mn-(m-1)$.
	\item
	$L_m(a,b,-,d) \cup L_n(b,a,-,d)$ has complexity $mn-(m+n-2)$.
	\item 
	If $m\neq n$, the bounds are met by $L_m(a,b,-,d)$ and $ L_n(a,b,-,d)$.
	\ee
\end{enumerate}
\end{theorem}

\begin{proof}
For $1\le q \le n-1$, a non-final state $q$ accepts $a^{n-1-q}d$ and no other non-final state accepts this word. All non-final states are distinguishable from the final state $n$. Hence $\cD_n(a,-,-,d)$ is minimal and $L_n(a,-,-,d)$ has $n$ quotients.
Since $\cD(a,b,c,d)$ has only one final state and that state accepts $\Sig^*$, it is a right ideal.
\be
\item
The case $n= 3$ is easily checked. For $n\ge 4$,
let $\cD'_n=(Q_n,\Sig,\delta'_n, 1,\{n\})$, where  $\Sig=\{a,b,c,d\}$, and
$a\colon(1,\ldots,n-1)$, $b\colon (1,2)$, ${c\colon(n-1 \rightarrow 1)}$ and 
${{d\colon(n-1\rightarrow n)}}$.
It was proved in~\cite{BrYe11} that  the transition semigroup of a minimal DFA accepting a right ideal has at most $n^{n-1}$ transformations, and that the transition semigroup of 
$\cD'_n$ has cardinality $n^{n-1}$. 
Since for $n\ge 3$,   $(1,2)$  is induced by  $a^{n-2}b$ in $\cD_n$,
all the transformations of $\cD'_n$ can be induced in $\cD_n$,
and the claim follows.
Moreover, it was proved in~\cite{BSY15} that an alphabet of at least four letters is required to meet this bound.

\item
Each quotient of $L_n(a,-,-,d)$, except $\{a,d\}^*$, has complexity $n$, since states $1,\ldots, n-1$ are strongly connected.
Each right ideal must have a final state that accepts $\Sig^*$ (for $L_n(a,-,-,d)$ this is state $n$), and so the complexity of the quotient corresponding to this final state is $1$.
Hence the complexities of the quotients are maximal for right ideals.
\item
It was proved  in~\cite{BrTa13} that the number of atoms of any regular language $L$ is equal to the complexity of the reverse of $L$.
If $L$ is a right ideal of complexity $n$, the maximal complexity of the reverse $L^R$ is 
$2^{n-1}$~\cite{BJL13}.
For $n = 3$, it is easily checked that our witness meets this bound.
For $n>3$, it was proved in~\cite{BrYe11} that the reverse of $L_n(a,-,-,d)$ reaches this bound.
\item
This was proved in~\cite{BrDa15}.
\item
See Item 3 above.
\item
See Theorem~\ref{thm:RStar}.
\item
See Theorem~\ref{thm:RProd}.
\item
See Theorems~\ref{thm:RBool} and \ref{thm:RBool2}.  
\ee
\end{proof}

\subsection{Star}
\begin{theorem}[Right Ideals: Star~\cite{BrDa14}]
\label{thm:RStar}
For $n \ge 3$ the complexity of the  star of $L_n(a,-,-,d)$ is $n+1$.
\end{theorem}
\begin{proof}
If $L$ is a right ideal, then $L^*=L\cup  \{\eps\} $.
Consider the DFA for star constructed as follows.
To add $\eps$ to $L$, 
the initial state $1$ of our witness cannot be made final since this would add other words to the language, for example, $a^{n-1}$.
Thus an additional state, say $0$, is required; this state is both initial and final in the DFA for $L^*$, its outgoing transitions are the same as those of the initial state $1$ of the DFA for $L$, and it has no incoming transitions. Since this DFA accepts $L\cup \{\eps\}$, $n+1$ states are sufficient. All the states are pairwise distinguishable, since $0$ rejects $d$, while $n$ accepts it, and all the non-final states are distinguishable by words in $a^*d$. 
\end{proof}

\subsection{Product}
\label{ssec:product_right}

We use the DFAs $\cD'_m(a,b,-,d)=(Q'_m,\Sig, \delta'_m,1',\{m'\})$ and $\cD_n(a,b,-,d)=(Q_n,\Sig, \delta_n,1,\{n\})$ shown in
Figure~\ref{fig:RProd} for $m=4$ and $n=5$, where $\Sig=\{a,b,d\}$ and the states of the first DFA are primed to distinguish them from those of the second DFA.

We show that the complexity of the product of $L_m(a,b,-,d)L_n(a,b,-,d)$ 
reaches the maximum possible bound $m+2^{n-2}$ derived in~\cite{BJL13}.
Define the $\eps$-NFA $\cP=(Q'_m\cup Q_n,\Sig, \delta_\cP,\{1'\},\{n\})$, 
where $\delta_\cP(p',a)=\{ \delta'(p',a)\}$ if $p'\in Q'_m$, $a\in\Sig$,
$\delta_\cP(q,a)=\{ \delta(q,a)\}$ if $q\in Q_n$, $a\in \Sig$, 
and $\delta_\cP(m',\eps)=\{1\}$.
This $\eps$-NFA accepts $L_mL_n$, and is illustrated in Figure~\ref{fig:RProd}.

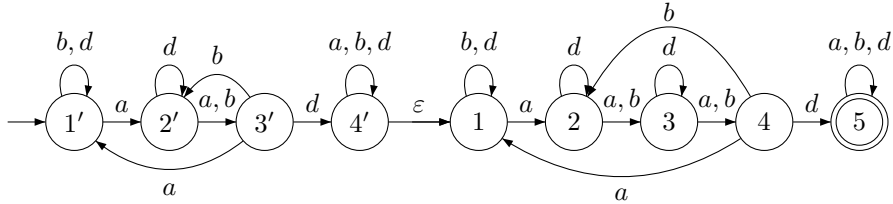
\begin{figure}[ht]
\unitlength 9pt
\begin{center}\begin{picture}(36,8)(0,1)
\gasset{Nh=2.4,Nw=2.4,Nmr=1.2,ELdist=0.4,loopdiam=1.2}
\node(1')(2,4){$1'$}\imark(1')
\node(2')(6,4){$2'$}
\node(3')(10,4){$3'$}
\node(4')(14,4){$4'$}
\drawedge(1',2'){$a$}
\drawedge(2',3'){$a,b$}
\drawedge(3',4'){$d$}

\drawedge[curvedepth=-2,ELdist=-1.2](3',2'){$b$}
\drawedge[curvedepth=2,ELdist=0.6](3',1'){$a$}

\drawloop(1'){$b,d$}
\drawloop(2'){$d$}
\drawloop(4'){$a,b,d$}

\node(1)(19,4){$1$}\imark(1)
\node(2)(23,4){$2$}
\node(3)(27,4){$3$}
\node(4)(31,4){$4$}
\node(5)(35,4){$5$}\rmark(5)

\drawedge(1,2){$a$}
\drawedge(2,3){$a,b$}
\drawedge(3,4){$a,b$}
\drawedge(4,5){$d$}

\drawedge[curvedepth=-4,ELdist=-1.0](4,2){$b$}
\drawedge[curvedepth=2.3,ELdist=0.5](4,1){$a$}

\drawedge(4',1){$\eps$}

\drawloop(1){$b,d$}
\drawloop(2){$d$}
\drawloop(3){$d$}
\drawloop(5){$a,b,d$}
\end{picture}\end{center}
\caption{$\eps$-NFA for product with $m=4$, $n=5$.}
\label{fig:RProd}
\end{figure}

\begin{theorem}[Right Ideals: Product~\cite{BrDa14}]
\label{thm:RProd}
For $m,n \ge 3$ the complexity of the product of $L_m(a,b,-,d)$ and $L_n(a,b,-,d)$ is $m+2^{n-2}$.
\end{theorem}

\begin{proof}
It was shown in~\cite{BJL13} that $m+2^{n-2}$ is an upper bound on the complexity of the product of two right ideals. To prove this bound is met, 
we apply the subset construction to $\cP$ to obtain a DFA $\cD$ for $L_mL_n$.
The states of $\cD$ are subsets of $Q'_m\cup Q_n$.
We prove that all states of the form $\{p'\}$, $p=1,\ldots,m-1$ and all states of the form
$\{m',1\}\cup S$, where $S\subseteq Q_n\setminus\{1,n\}$, and state 
$\{m',1,n\}$ are reachable, for a total of $m+2^{n-2}$ states.

State $\{1'\}$ is the initial state, and  $\{p'\}$ is reached by $a^{p-1}$ for $p=2,\ldots,m-1$.
Also, $\{m',1\}$ is reached by $a^{m-2}d$, and  
states $m'$ and 1 are present in every subset reachable from $\{m',1\}$. 
By applying $ab^{q-2}$ to $\{m',1\}$ we reach $\{m',1,q\}$; hence all subsets 
$\{m',1\}\cup S$ with $|S|=1$ are reachable.

Assume now that we can reach all sets $\{m',1\}\cup S$ with $|S|=k$, and 
suppose that we want to reach $\{m',1\}\cup T$ with $T=\{q_0,q_1,\ldots,q_k\}$
with $2\le q_0<q_1<\cdots <q_k\le n-1$.
 This can be done by starting with $S=\{q_1-q_0+1, \ldots, q_k-q_0+1\}$ and applying $ab^{q_0-2}$. 
Finally, to reach $\{m',1,n\}$, apply $d$ to $\{m',1,n-1\}$.

If $1\le  p < q\le m-1$, then state $\{p'\}$ is distinguishable from $\{q'\}$ by 
$a^{m-1-q}da^{n-1}d$.
Also, state $p \in Q_n$ with $2\le p \le n-1$ accepts $a^{n-1-p}d$ and no other state $q\in Q_n$ with $2\le q\le n-1$ accepts this word. 
Hence, if $S,T \subseteq Q_n\setminus\{1,n\}$ and $S\ne T$, then 
$\{m',1\}\cup S$ and $\{m',1\}\cup T$ are distinguishable.
State $\{p'\}$ with $2\le p\le m-1$ is distinguishable from state $\{m',1\}\cup S$
because there is a word with a single $d$ that is accepted from $\{m',1\}\cup S$
but no such word is accepted by $\{q'\}$. Hence all the non-final states are distinguishable, and $\{m',1,n\}$ is the only final state.
\end{proof}

\subsection{Boolean Operations}
\label{ssec:boolean_right}

We restrict our attention to the four boolean operations $\cup,\cap, \setminus,\oplus$, since the complexity of any other proper binary boolean operation can be obtained from these four. For example,
we have $\kappa(L' \cup\ol{L}) = \kappa \ol{ (L' \cup\ol{L} ) } = \kappa ( \ol{L'} \cap L) = \kappa(L\setminus L')$.

Tight upper bounds for boolean operations on right ideals~\cite{BJL13}  are $mn$ for intersection and symmetric difference, $mn-(m-1)$ for difference, and $mn-(m+n-2)$ for union.
Since $L_n\cup L_n=L_n\cap L_n=L_n$, and 
$L_n\setminus L_n=L_n \oplus L_n=\emp$, two different languages must be used
to reach the bounds if $m=n$. 

We use the DFAs $\cD'_m(a,b,-,d)=(Q'_m,\Sig, \delta'_m,1',\{m'\})$ and $\cD_n(b,a,-,d)=(Q_n,\Sig, \delta_n,1,\{n\})$ shown in
Figure~\ref{fig:RBool} for $m=4$ and $n=5$.
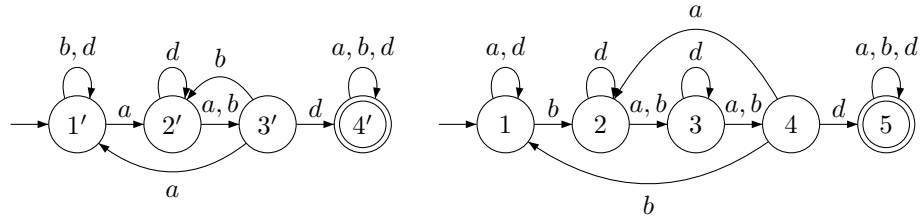
\begin{figure}[ht]
\unitlength 9pt
\begin{center}\begin{picture}(36,8)(0,1)
\gasset{Nh=2.4,Nw=2.4,Nmr=1.2,ELdist=0.3,loopdiam=1.2}
\node(1')(2,4){$1'$}\imark(1')
\node(2')(6,4){$2'$}
\node(3')(10,4){$3'$}
\node(4')(14,4){$4'$}\rmark(4')
\drawedge(1',2'){$a$}
\drawedge(2',3'){$a,b$}
\drawedge(3',4'){$d$}

\drawedge[curvedepth=-2,ELdist=-1.2](3',2'){$b$}
\drawedge[curvedepth=2,ELdist=0.6](3',1'){$a$}

\drawloop(1'){$b,d$}
\drawloop(2'){$d$}
\drawloop(4'){$a,b,d$}

\node(1)(20,4){$1$}\imark(1)
\node(2)(24,4){$2$}
\node(3)(28,4){$3$}
\node(4)(32,4){$4$}
\node(5)(36,4){$5$}\rmark(5)

\drawedge(1,2){$b$}
\drawedge(2,3){$a,b$}
\drawedge(3,4){$a,b$}
\drawedge(4,5){$d$}

\drawedge[curvedepth=-4,ELdist=-1.0](4,2){$a$}
\drawedge[curvedepth=2.5,ELdist=0.5](4,1){$b$}

\drawloop(1){$a,d$}
\drawloop(2){$d$}
\drawloop(3){$d$}
\drawloop(5){$a,b,d$}
\end{picture}\end{center}
\caption{Right-ideal witnesses for boolean operations with $m=4$, $n=5$.}
\label{fig:RBool}
\end{figure}

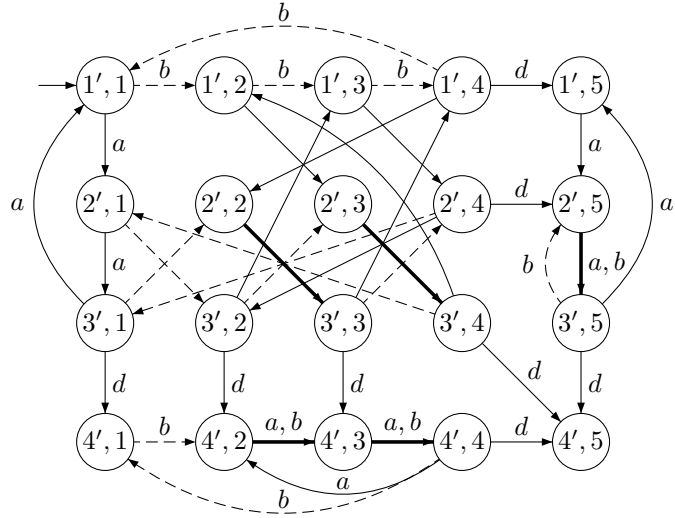
\begin{figure}[ht]
\unitlength 9pt
\begin{center}\begin{picture}(25,20)(0,-2)
\gasset{Nh=2.4,Nw=2.4,Nmr=1.2,ELdist=0.3,loopdiam=1.2}
\node(1'1)(2,15){$1',1$}\imark(1'1)
\node(2'1)(2,10){$2',1$}
\node(3'1)(2,5){$3',1$}
\node(4'1)(2,0){$4',1$}

\node(1'2)(7,15){$1',2$}
\node(2'2)(7,10){$2',2$}
\node(3'2)(7,5){$3',2$}
\node(4'2)(7,0){$4',2$}

\node(1'3)(12,15){$1',3$}
\node(2'3)(12,10){$2',3$}
\node(3'3)(12,5){$3',3$}
\node(4'3)(12,0){$4',3$}

\node(1'4)(17,15){$1',4$}
\node(2'4)(17,10){$2',4$}
\node(3'4)(17,5){$3',4$}
\node(4'4)(17,0){$4',4$}

\node(1'5)(22,15){$1',5$}
\node(2'5)(22,10){$2',5$}
\node(3'5)(22,5){$3',5$}
\node(4'5)(22,0){$4',5$}

\drawedge(1'1,2'1){$a$}
\drawedge(2'1,3'1){$a$}
\drawedge(3'1,4'1){$d$}

\drawedge[curvedepth=3,ELdist=.4](3'1,1'1){$a$}

\drawedge(1'2,2'3){}
\drawedge[linewidth=0.15](2'2,3'3){}
\drawedge(3'2,1'3){}
\drawedge(3'2,4'2){$d$}

\drawedge(1'3,2'4){}
\drawedge[linewidth=0.15](2'3,3'4){}
\drawedge(3'3,1'4){}
\drawedge(3'3,4'3){$d$}

\drawedge(1'4,2'2){}
\drawedge(2'4,3'2){}
\drawedge(3'4,4'5){$d$}

\drawedge(1'5,2'5){$a$}
\drawedge[linewidth=0.15](2'5,3'5){$a,b$}
\drawedge(3'5,4'5){$d$}

\drawedge[curvedepth=-3,ELdist=-0.9](3'5,1'5){$a$}

\drawedge[curvedepth=-2,ELdist=-0.9](3'4,1'2){}

\drawedge[dash={.5 .25}{.25}](4'1,4'2){$b$}
\drawedge[linewidth=0.15](4'2,4'3){$a,b$}
\drawedge[linewidth=0.15](4'3,4'4){$a,b$}
\drawedge(4'4,4'5){$d$}

\drawedge[curvedepth=2.2,ELdist=-.7](4'4,4'2){$a$}
\drawedge[curvedepth=3,ELdist=-1,dash={.5 .25}{.25}](4'4,4'1){$b$}
\drawedge[dash={.5 .25}{.25}](1'1,1'2){$b$}
\drawedge[dash={.5 .25}{.25}](1'2,1'3){$b$}
\drawedge[dash={.5 .25}{.25}](1'3,1'4){$b$}
\drawedge(1'4,1'5){$d$}
\drawedge[curvedepth=-2.5,ELdist=-1,dash={.5 .25}{.25}](1'4,1'1){$b$}

\drawedge[dash={.5 .25}{.25}](2'1,3'2){}
\drawedge[dash={.5 .25}{.25}](3'1,2'2){}
\drawedge[dash={.5 .25}{.25}](3'2,2'3){}
\drawedge[curvedepth=1.5,ELdist=.5,dash={.5 .25}{.25}](3'5,2'5){$b$}

\drawedge[dash={.5 .25}{.25}](2'4,3'1){}
\drawedge[dash={.5 .25}{.25}](3'3,2'4){}
\drawedge[dash={.5 .25}{.25}](3'4,2'1){}
\drawedge(2'4,2'5){$d$}

\end{picture}\end{center}
\caption{Direct product for boolean operations. Unlabeled transitions under $a$ are solid, under $b$ are dashed, and under both $a$ and $b$ are thick.   Self-loops are omitted.}
\label{fig:RCross}
\end{figure}

Let $\cD_{m,n} = \cD'_m \times \cD_n$.
Depending on the assignment of final states, this DFA recognizes different boolean operations on $L_m$ and $L_n$.
The 
direct product  of $\cD_4(a,b,-,d)$ and  $\cD_5(b,a,-,d)$ is  in
Figure~\ref{fig:RCross}.

Let $S_n$ denote the symmetric group of degree $n$. 
A \emph{basis}~\cite{Pic39} of $S_n$
is an ordered pair $(s,t)$ of distinct transformations of $Q_n=\{1,\dots,n\}$ that generate $S_n$.
A DFA \emph{has a basis $(t_a,t_b)$ for $S_n$} if it has letters $a,b\in \Sig$ such that $a$ induces $t_a$ and $b$ induces $t_b$.
In the case of our right ideal $\cD'_m$ ($\cD_n$), the transitions $\delta'_a$ and $\delta'_b$ ($\delta_a$ and $\delta_b$) restricted to $\{1',\dots,(m-1)'\}$($\{1,\dots,n-1\}$), constitute a basis for $S_{m-1}$ ($S_{n-1})$. 
Consider DFAs 
$\cA'_{m-1}=(Q'_{m-1},\Sig,\delta'_{m-1},1',\{(m-1)'\})$ and 
$\cA_{n-1}=(Q_{n-1},\Sig,\delta_{n-1}, 1, \{n-1\})$, where $\Sig=\{a,b\}$ and the transitions of $\cA'_{m-1}$ ($\cA_{n-1}$) are the same as those of $\cD'_m$ ($\cD_n)$ restricted to $Q'_{m-1}$ ($Q_{n-1})$.
By~\cite[Theorem 1]{BBMR14} all the states in the direct product of 
$\cA'_{m-1}$  and $\cA_{n-1}$ are reachable and distinguishable for 
$m-1,n-1\ge 2$,  $(m-1,n-1)\not \in \{(2,2), (3,4),(4,3),(4,4)\}$.
We shall use this result to simplify our proof of the next theorem.

\begin{theorem}[Right Ideals: Boolean Operations~\cite{BrDa14}]
\label{thm:RBool}
If $m,n\ge 3$, then
\be
\item
The complexity of $L_m(a,b,-,d) \cap L_n(b,a,-,d)$ is $mn$.
\item
The complexity of $L_m(a,b,-,d) \oplus L_n(b,a,-,d)$ is $mn$.
\item
The complexity of $L_m(a,b,-,d) \setminus L_n(b,a,-,d)$ is $mn-(m-1)$.
\item
The complexity of $L_m(a,b,-,d) \cup L_n(b,a,-,d)$ is $mn-(m+n-2)$.
\ee
\end{theorem}

\begin{proof}
In the cases where $(m,n) \in \{(3,3),(4,5),(5,4),(5,5)\}$, we cannot 
apply the result from~\cite[Theorem 1]{BBMR14}, but we have verified computationally that the bounds are met. Thus we assume that $(m,n) \not\in \{(3,3),(4,5),(5,4),(5,5)\}$.

Our first task is to show that all $mn$ states of $\cD_{m,n}$ are reachable.
By~\cite[Theorem 1]{BBMR14}, all states in the set $S = \{(p',q) \mid 1 \le p \le m-1, 1 \le q \le n-1\}$ are reachable. The remaining states are the ones in the last row or last column (that is, Row $m$ or Column $n$) of the direct product.

For $1 \le q \le n-2$, from state $((m-1)',q)$ we can reach $(m',q)$ by $d$. From state $(m',n-2)$ we can reach $(m',n-1)$ by $a$. From state $((m-1)',n-1)$ we can reach $(m',n)$ by $d$. Hence all states in Row $m$ are reachable.

For $1 \le p \le m-2$, from state $(p',n-1)$ we can reach $(p',n)$ by $d$. From state $((m-2)',n)$ we can reach $((m-1)',n)$ by $a$. Hence all states in Column $n$ are reachable, and thus all $mn$ states are reachable.

We now count the number of distinguishable states for each operation. Let $H = \{(m',q) \mid 1 \le q \le n\}$ be the set of states in the last row and let $V = \{(p',n) \mid 1 \le p \le m\}$ be the set of states in the last column. If $\circ \in \{\cap,\oplus,\setminus,\cup\}$, then $L_m(a,b,-,d) \circ L_n(b,a,-,d)$ is recognized by $\cD_{m,n}$, where the set of final states is taken to be $H \circ V$.

Let $H' = \{((m-1)',q) \mid 1 \le q \le n-1\}$ and let $V' = \{(p',n-1) \mid 1 \le p \le m-1\}$.
That is, $H'$ is the second last row of states, and $V'$ is the second last column, restricted to $S$. 
By~\cite[Theorem 1]{BBMR14}, all states in $S$ are distinguishable with respect to $H' \circ V'$, for each boolean operation $\circ \in \{\cap,\oplus,\setminus,\cup\}$.
We claim that they are also distinguishable with respect to $H \circ V$ for $\circ \in \{\cap,\oplus,\setminus,\cup\}$.

To see this, one verifies the following statement: for each $\circ \in \{\cap,\oplus,\setminus,\cup\}$ and each state $(p',q) \in S$, we have $(p',q) \in H' \circ V'$ if and only if $(p',q)d \in H \circ V$. (This \emph{only} applies to states $(p',q) \in S$; for example, in Figure \ref{fig:RCross} we see that $(4',4)d \in H \cap V$ but $(4',4) \not \in  H' \circ V'$.) Since states in $S$ are distinguishable with respect to $ H' \circ V'$, it follows that for any pair of states $(p',q),(r',s) \in S$ there is a word $w$ with $(p',q)w \in H' \circ V'$ and $(r',s) \not\in  H' \circ V'$. Then by the statement, the word $wd$ sends $(p',q)$ into $H \circ V$ and $(r',s)$ outside of $H \circ V$, thus distinguishing the two states with respect to $H \circ V$.

Thus for each boolean operation $\circ$, all 
states in $S$ are distinguishable from each other with respect to the final state set $H \circ V$. Next, we prove that the states in $S$ are distinguishable from the rest of the states (those in $H \cup V$) with respect to the final state set $H \circ V$.

Since all states in $S$ are non-final, it suffices to distinguish states in $S$ from states in $X = H \cup V \setminus (H \circ V)$, the set of non-final states in $H \cup V$. If $\circ = \cup$, then $X$ contains no states and there is nothing to be done. 
If $\circ = \oplus$, the only state in $X$ is $(m',n)$, which is empty, but all states in $S$ are non-empty. If $\circ = \setminus$, then all states in $X$ are empty but no states in $S$ are.
Finally, if $\circ = \cap$, observe that each state in $X$ accepts a word containing a single $d$, while states in $S \setminus \{((m-1)',n-1)\}$  accept only words with at least two occurrences of $d$. To distinguish $((m-1)',n-1)$ from $(p',q) \in X$, apply $a$ (which maps $((m-1)',n-1)$ to $(1',2)$) and then apply a word which distinguishes $(1',2)$ from $(p',q)a$.

We have shown that each state in $S$ is distinguishable from every other state in the direct product $\cD_{m,n}$, with respect to each of the four final state sets $H \circ V$ with $\circ \in \{\cap,\oplus,\setminus,\cup\}$.
Since there are $(m-1)(n-1) = mn - m - n + 1$ states in $S$,  there are at least that many distinguishable states for each operation $\circ$.
To show that the complexity bounds are reached by $L_m(a,b,-,d) \circ L_n(b,a,-,d)$, it suffices to consider how many of the remaining $m + n - 1$ states in $H \cup V$ are distinguishable with respect to $H \circ V$.
We consider each operation in turn.

\noin\textbf{Intersection:}
Here the set of final states is $H \cap V = \{(m',n)\}$.
State $(m',n)$ is the only final state and hence is distinguishable from all the other states.
Any two states in $H$ ($V$) are distinguished by words in $b^*d$ ($a^*d$).
State $(m',1)$ accepts $b^{n-2}d$, while $(1',n)$ rejects it.
For $2\le q\le n-1$, $(m',q)$ is sent to $(m',1)$ by $b^{n-q}$, while 
state $(1',n)$ is not changed by that word.
Hence $(m',q)$ is distinguishable from $(1',n)$.
By a symmetric argument, $(p',n)$ is distinguishable from $(m',1)$ for $2 \le p \le m-1$.
For $2\le q\le n-1$ and $2\le p\le m-1$, $(m',q)$ is distinguished from
$(p',n)$ because $b^{n-q}$ sends the former to $(m',1)$ and the latter to a state of the form $(r',n)$, where $2\le r \le m-1$.
Hence all pairs of states from $H \cup V$ are distinguishable. 
Since there are $m + n - 1$ states in $H \cup V$,  it follows there are $(mn - m - n + 1) + (m + n - 1) = mn$ distinguishable states.

\noin\textbf{Symmetric Difference:}
Here the set of final states is $H \oplus V$, that is, all states in the last row and column except $(m',n)$, which is the only empty state.
This situation is complementary to that for intersection. Thus every two states from $H \cup V$ are distinguishable by the same word as for intersection.
Hence there are $mn$ distinguishable states.

\noin\textbf{Difference:}
Here the set of final states is $H \setminus V$, that is, all states in the last row $H$ except $(m',n)$, which is empty. All other states in the last column $V$ are also empty. 
The $m$ empty states in $V$ are all equivalent, and the $n-1$ final states in $H \setminus V$ are distinguished in the same way as for intersection. Hence there are $(n-1)+1 = n$ distinguishable states in $H \setminus V$. It follows there are $(mn - m - n + 1) + n = mn - (m-1)$ distinguishable states.

\noin\textbf{Union:}
Here the set of final states is $H \cup V$. From a state in $H \cup V$ it is possible to reach only other states in $H \cup V$, and all these states are final; so every state in $H \cup V$ accepts $\Sig^*$. Thus all the states in $H \cup V$ are equivalent, and so there are $(mn - m - n + 1) + 1 = mn - (m + n - 2)$ distinguishable states.
\end{proof}

Although it is impossible for the stream $(L_n(a,b,-,d) \mid n\ge 3)$ to meet the bounds for boolean operations when $m=n$, this stream is as complex as it could possibly be
in view of the following theorem:

\begin{theorem}[Right Ideals: Boolean Operations, $m\neq n$]
\label{thm:RBool2}
Suppose $m,n \ge 3$ and $m \ne n$.
\be
\item
The complexity of $L_m(a,b,-,d) \cap L_n(a,b,-,d)$ is $mn$.
\item
The complexity of $L_m(a,b,-,d) \oplus L_n(a,b,-,d)$ is $mn$.
\item
The complexity of $L_m(a,b,-,d) \setminus L_n(a,b,-,d)$ is $mn-(m-1)$.
\item
The complexity of $L_m(a,b,-,d) \cup L_n(a,b,-,d)$ is $mn-(m+n-2)$.
\ee
\end{theorem}
\begin{proof}
Let $\cD'_m = \cD_m(a,b,-,d)$, $\cD_n = \cD_n(a,b,-,d)$, and $\cD_{m,n} = \cD'_m \times \cD_n$ be the direct product automaton. 
If $(m,n) \in \{(4,5),(5,4)\}$, we have verified computationally that the bounds are met. If $(m,n) \not \in \{(4,5),(5,4)\}$,  we can apply~\cite[Theorem 1]{BBMR14}. Thus by the arguments used in the proof of Theorem \ref{thm:RBool}, all states of $\cD_{m,n}$ are reachable. 

Most of the distinguishability arguments carry over as well.
If $H$ and $V$ are the last row and column of states in $\cD_{m.n}$ respectively, and $S$ is the set of states lying outside of $H \cup V$, we can use nearly identical arguments as in the proof of Theorem \ref{thm:RBool} to show that for $\circ \in \{\cap,\oplus,\setminus,\cup\}$, every state in $S$ is distinguishable from every other state in $\cD_{m,n}$ with respect to $H \circ V$.
It remains to count the number of states in $H \cup V$ that are distinguishable with respect to $H \circ V$.

\noin\textbf{Intersection:}
Here the set of final states is $H \cap V = \{(m',n)\}$. Since $(m',n)$ is the only final state, it is distinguishable from all other states.
Any two states both in $H$ (or both in $V$) are distinguished by words in $a^*d$.
Suppose $m < n$. Then $a^{m-1}$ sends $(m',1)$ to $(m',m)$ and fixes $(1',n)$.
Words in $b^*$ can send $(m',m)$ to $(m',q)$ for $2 \le q \le n-1$, and they fix $(1',n)$.
For $2 \le q \le n-1$, $(m',q)$ accepts $b^{n-1-q}d$, while $(1',n)$ remains fixed.
Hence $(m',q)$ is distinguishable from $(1',n)$ for all $q$. 
For $2 \le p \le m-1$ and $2 \le q \le n-1$, $(m',q)$ is distinguished from $(p',n)$ because $a^{m-p}$ sends $(p',n)$ to $(1',n)$ and $(m',q)$ to some state that is distinguishable from $(1',n)$.
Hence all pairs of states from $H \cup V$ are distinguishable if $m < n$. A symmetric argument works for $m > n$. Thus all $mn$ states are distinguishable.

\noin\textbf{Symmetric Difference, Difference, and Union:}
The arguments are similar to those used in the proof of Theorem \ref{thm:RBool}. 
\end{proof}

\begin{remark}
For each class of languages we studied in this paper, our goal was to find a single DFA stream that meets the upper bounds (for that class) on all of our complexity measures.
For regular right ideals, a four-letter alphabet was necessary to achieve this, because fewer than four letters are not sufficient for the size of the syntactic semigroup to be maximal. Having found such a DFA, we then observed that the alphabet of this DFA can be reduced for several operations. On the other hand, if one wishes to minimize the alphabet for one particular operation only, it is possible to find witnesses over even smaller alphabets. 

We list here each operation with the size of the smallest known alphabet (first entry) along with our alphabet size (second entry):
reversal (2/2), star (2/2), product (2/3), union (2/3), intersection (2/3), symmetric difference (2/3), and difference (2/3).

As an example, consider the two binary witnesses for the product operation that are used in~\cite{BJL13}: 
$\cD'_m=(Q'_m,\{a,b\},\delta'_m, 1', \{m'\})$, 
where 
$a,b \colon ((m-1)' \to m')((m-2)' \to (m-1)') \cdots (1' \to 2'),$
and 
$\cD_n=(Q_n,\{a,b\},\delta_n, 1, \{n\})$, 
where 
$ a\colon (1,2,\dots, n-1)$, $b \colon (n-1 \to n)(n-2 \to n-1) \cdots (2 \to 3).$
Note that, although only two inputs are used, they induce three different transformations. 
Thus one can argue that these witnesses are not simpler than ours.
\end{remark}

\section{Left Ideals}

The following stream of left ideals was  defined in~\cite{BrYe11}:

\begin{definition}
\label{def:LWit}
For $n\ge 4$, let $\cD_n=\cD_n(a,b,c,d,e)=(Q_n,\Sig,\delta_n, 1, \{n\})$, where 
$\Sig=\{a,b,c,d,e\}$,
and $\delta_n$ is defined by  transformations
$a\colon (2,\dots,n)$,
$b\colon(2,3)$,
${c\colon(n \to 2)}$,
${d\colon(n\to 1)}$, 
$e\colon (Q_n\to 2)$.
Let $L_n=L_n(a,b,c,d,e)$ be the language accepted by~$\cD_n$.
The structure of  $\cD_n(a,b,c,d,e)$ is shown in Figure~\ref{fig:LWit}. 
\end{definition}

\begin{figure}[ht]
\unitlength 10.5pt
\begin{center}\begin{picture}(33,11)(-.5,2.0)
\gasset{Nh=2.4,Nw=2.4,Nmr=1.2,ELdist=0.3,loopdiam=1.2}
\node(1)(2,8){$1$}\imark(1)
\node(2)(7,8){$2$}
\node(3)(12,8){$3$}
\node(4)(17,8){$4$}
\node[Nframe=n](qdots)(22,8){$\dots$}
{\scriptsize
\node(n-1)(27,8){{\small $n-1$}}
\node(n)(31,8){{\small $n$}}\rmark(n)
}
\drawedge(1,2){$e$}
\drawedge[ELdist=.2](2,3){$a,b$}
\drawedge(3,4){$a$}
\drawedge(4,qdots){$a$}
\drawedge(qdots,n-1){$a$}
\drawedge(n-1,n){$a$}
\drawloop(1){$a,b,c,d$}
\drawloop(2){$c,d,e$}
\drawloop(3){$c,d$}
\drawloop(4){$b,c,d$}
\drawloop(n-1){$b,c,d$}
\drawloop(n){$b$}
\drawedge[curvedepth=-2.5](3,2){$b,e$}
\drawedge[curvedepth=-4.8](4,2){$e$}
\drawedge[curvedepth=2.1](n-1,2){$e$}
\drawedge[curvedepth=3.5](n,2){$a,c,e$}
\drawedge[curvedepth=5](n,1){$d$}
\end{picture}\end{center}
\caption{Minimal DFA $\cD_n(a,b,c,d,e)$  of Definition~\ref{def:LWit}.}
\label{fig:LWit}
\end{figure}
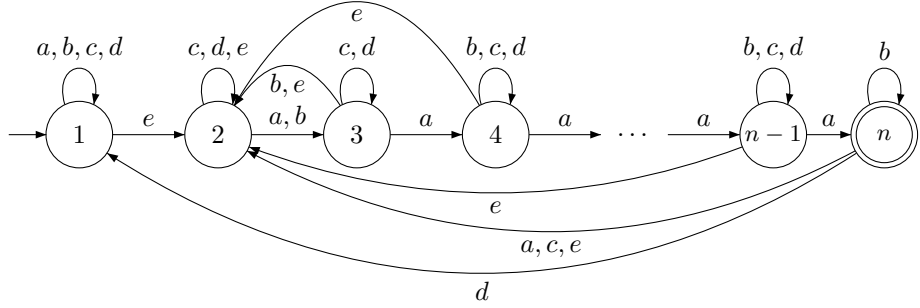

This stream of languages is closely related to the stream of Figure~\ref{fig:RegWit}.
The DFA $\cD_n(a,b,c,d,e)$ of Definition~\ref{def:LWit} is constructed by taking $\cD_{n-1}(a,b,c)$ of Figure~\ref{fig:RegWit} with states relabeled $2,\dots,n$, adding a new state $1$ and new inputs $d\colon (n \rightarrow 1)$ and $e\colon (Q_n \to 2)$.

\begin{theorem}[Left Ideals]
For each $n\ge 4$, the DFA $\cD_n(a,b,c,d,e)$ of Definition~\ref{def:LWit} is minimal and its 
language $L_n(a,b,c,d,e)$ is a left ideal of complexity $n$.
The stream $(L_n(a,b,c,d,e) \mid n \ge 4)$  with dialect stream
$(L_n(a,b,e,d,c) \mid n \ge 4)$
is most complex in the class of regular left ideals.
In particular, this stream meets all the complexity bounds listed below, which are maximal for left ideals. In several cases the bounds can be met with restricted alphabets, as shown below.
\be
\item
The syntactic semigroup of $L_n(a,b,c,d,e)$ has cardinality $n^{n-1}+n-1$.  Moreover, fewer than five inputs do not suffice to meet this bound.
\item
All quotients of $L_n(a,-,-,d,e)$ have complexity $n$.
\item
$L_n(a,-,c,d,e)$ has $2^{n-1}+1$ atoms. 
\item
For each atom $A_S$ of $L_n(a,b,c,d,e)$, the complexity $\kappa(A_S)$ satisfies:
\begin{equation*}
	\kappa(A_S) =
	\begin{cases}
		 n, 			& \text{if $S=Q_n$;}\\
		2^{n-1},		& \text{if $S=\emp$;}\\
		1 + \sum_{x=1}^{|S|}\sum_{y=1}^{n-|S|}\binom{n-1}{x}\binom{n-x-1}{y-1},
		 			& \text{otherwise.}
		\end{cases}
\end{equation*}
\item
The reverse of $L_n(a,-,c,d,e)$ has complexity $2^{n-1}+1$.

\item
The star of $L_n(a,-,-,-,e)$ has complexity $n+1$.
\item
The product $L_m(a,-,-,-,e) L_n(a,-,-,-,e)$ has complexity $m+n-1$.
\item
For any proper binary boolean function $\circ$, the complexity of  $L_m(a,-,c,-,e)\\ \circ L_n(a,-,e,-,c)$
is $mn$.

\ee
\end{theorem}
\begin{proof}
DFA $\cD_n(a,-,-,d,e)$ is minimal  since 
only state 1 does not accept any word in $a^*$, whereas every other state $p >1$ accepts $a^{n-p }$ and no state $q$ with $1<q\neq p$  accepts this word.
It was proved in~\cite{BrYe11} that $\cD_n(a,b,c,d,e)$  accepts a left ideal.
\be
\item
It was shown in~\cite{BrSz14} that the syntactic semigroup of a left ideal of complexity $n$ has cardinality at most $n^{n-1}+n-1$, and
in~\cite{BrYe11} that the syntactic semigroup of $L_n(a,b,c,d,e)$ has cardinality $n^{n-1}+n-1$. 
Moreover, it was proved in~\cite{BSY15} that an alphabet of at least five letters is required to meet this bound. 
\item
Each quotient of $L_n(a,-,-,d,e)$ has complexity $n$, since its minimal DFA is strongly connected.
\item
The number of atoms of any regular language $L$ is equal to the complexity of the reverse of $L$~\cite{BrTa13}. It was proved in~\cite{BrYe11} that the complexity of the reverse of $L_n(a,-,c,d,e)$ is $2^{n-1}+1$.
\item
This was proved in~\cite{BrDa15}.
\item
See Item~3 above.
\item
The argument is the same as for the star of right ideals.
\item
See Theorem~\ref{thm:LProd}.
\item
See Theorem~\ref{thm:LBool}.
\ee
\end{proof}

\subsection{Product}
\label{ssec:product_left}
We now show that the complexity of the product of $\cD'_m(a,-,-,-,e)$ with $\cD_n(a,-,-,-,e)$ 
reaches the maximum possible bound $m+n-1$ derived in~\cite{BJL13}. 
As in~\cite{BJL13} we use the following construction:
Define a DFA $\cD$ from DFAs $\cD'_m$ and $\cD_n$
by omitting  the final state of $\cD'_m$
and all the transitions from the final state,
and redirecting all the transitions that  go
from a non-final state of $\cD'_m$
to the final state of $\cD'_m$
to the initial state of $\cD_n$.
It was proved in~\cite{BJL13} that $\cD$ accepts $L_mL_n$.
The construction is illustrated in Figure~\ref{fig:LProd} for $m=4$ and $n=5$.

\begin{figure}[ht]
\unitlength 9pt
\begin{center}\begin{picture}(36,8)(0,1)
\gasset{Nh=2.4,Nw=2.4,Nmr=1.2,ELdist=0.4,loopdiam=1.0}
\node(1')(2,4){$1'$}\imark(1')
\node(2')(6,4){$2'$}
\node(3')(10,4){$3'$}
\drawedge(1',2'){$e$}
\drawedge(2',3'){$a$}

\drawedge[curvedepth=1.2,ELdist=.4](3',2'){$e$}

\drawloop(1'){$a$}
\drawloop(2'){$e$}

\node(1)(16,4){$1$}
\node(2)(20,4){$2$}
\node(3)(24,4){$3$}
\node(4)(28,4){$4$}
\node(5)(32,4){$5$}\rmark(5)

\drawedge(1,2){$e$}
\drawedge(2,3){$a$}
\drawedge(3,4){$a$}
\drawedge(4,5){$a$}

\drawedge[curvedepth= 1.2,ELdist=-.7](3,2){$e$}
\drawedge[curvedepth=-2.5,ELdist=-1.0](4,2){$e$}
\drawedge[curvedepth=3,ELdist=-1](5,2){$a,e$}

\drawedge(3',1){$a$}

\drawloop(1){$a$}
\drawloop(2){$e$}
\
\end{picture}\end{center}
\caption{Product of left-ideal witnesses with $m=4$, $n=5$.}
\label{fig:LProd}
\end{figure}
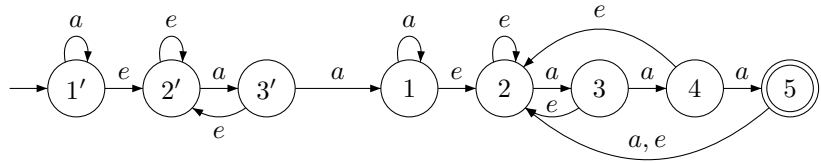

\begin{theorem} [Left Ideals: Product]
\label{thm:LProd}
For $m,n\ge 4$, the complexity of  the product of $L_m(a,-,-,-,e)$ and  $L_n(a,-,-,-,e)$ is $m+n-1$.
\end{theorem}
\begin{proof}
By construction $\cD$ has $m+n-1$ states.
It is also clear that the shortest word accepted by state $1'$ is $ea^{m-2}ea^{n-2}$, whereas for a state $p'$ with $2\le p\le m-1$ it is $a^{m-p}ea^{n-2}$, for state $1$ it is $ea^{n-2}$, and for any state $q$ with $2\le q\le n$ it is $a^{n-q}$.
Hence all the states are distinguishable by their shortest words.
\end{proof}

\subsection{Boolean Operations}
\label{ssec:boolean_left}
As pointed out earlier, two different languages have to be used to reach the maximal complexity for boolean operations. 
Let $\cD'_m = \cD'_m(a,-,c,-,e)$,  $\cD_n = \cD_n(a,-,e,-,c)$,
and $\cD_{m,n} = \cD'_m \times \cD_n $. Figure~\ref{fig:LBool} shows DFAs $\cD'_4(a,-,c,-,e)$ and  $\cD_5(a,-,e,-,c)$.

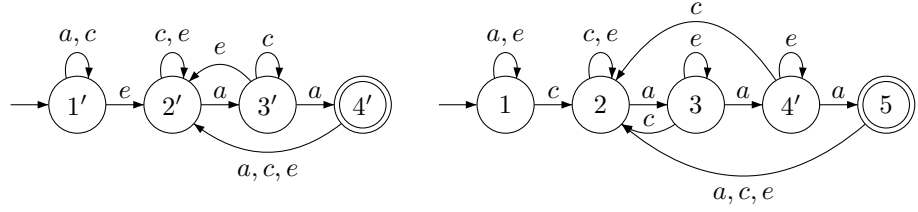
\begin{figure}[ht]
\unitlength 9pt
\begin{center}\begin{picture}(36,8)(0,1)
\gasset{Nh=2.4,Nw=2.4,Nmr=1.2,ELdist=0.3,loopdiam=1.0}
\node(1')(2,4){$1'$}\imark(1')
\node(2')(6,4){$2'$}
\node(3')(10,4){$3'$}
\node(4')(14,4){$4'$}\rmark(4')
\drawedge(1',2'){$e$}
\drawedge(2',3'){$a$}
\drawedge(3',4'){$a$}

\drawedge[curvedepth=-1.7,ELdist=-.9](3',2'){$e$}
\drawedge[curvedepth=2,ELdist=0.5](4',2'){$a,c,e$}

\drawloop(1'){$a,c$}
\drawloop(2'){$c,e$}
\drawloop(3'){$c$}

\node(1)(20,4){$1$}\imark(1)
\node(2)(24,4){$2$}
\node(3)(28,4){$3$}
\node(4)(32,4){$4'$}
\node(5)(36,4){$5$}\rmark(5)

\drawedge(1,2){$c$}
\drawedge(2,3){$a$}
\drawedge(3,4){$a$}
\drawedge(4,5){$a$}

\drawedge[curvedepth=1.2,ELdist=-.8](3,2){$c$}
\drawedge[curvedepth=-3.5,ELdist=-.8](4,2){$c$}
\drawedge[curvedepth=3.0,ELdist=0.5](5,2){$a,c,e$}

\drawloop(1){$a,e$}
\drawloop(2){$c,e$}
\drawloop(3){$e$}
\drawloop(4){$e$}
\end{picture}\end{center}
\caption{Left-ideal witnesses for boolean operations with $m=4$, $n=5$.}
\label{fig:LBool}
\end{figure}

\begin{theorem}[Left Ideals: Boolean Operations]
\label{thm:LBool}
If $m,n\ge 4$ and $\circ$ is any proper binary boolean function, then
the complexity of $L_m(a,-,c,-,e) \circ L_n(a,-,e,-,c)$ is $mn$.
\end{theorem}
\begin{proof}
Our first task is to show that all $mn$ states of $\cD_{m,n}$ are reachable.
State $(1',1)$ is the initial state. For $q=2,\dots,n$,  $(1',q)$ is reachable by $ca^{q-2}$, and for $p=2,\dots,m$, $(p',1)$ is reachable by $ea^{p-2}$.
Next,  $(p',2)$ is reached from $(p',1)$ by $c$ for $p=2,\dots, m-1$, and 
$(2',q)$ is reached from $(1',q)$ by $e$ for $q=2,\dots,n-1$.

Let $g=\gcd(m-1,n-1)$ and $\ell=\lcm(m-1,n-1)$; then $g\cdot \ell=(m-1)(n-1)$.
Note that $a$ is a permutation on the set $S=\{(p',q) \mid 2\le p\le m, 2\le q\le n\}$. Since $a$ is a cyclic permutation of order $m-1$ on $Q'_{m}\setminus\{1'\}$, and $a$ is also a cyclic permutation of order $n-1$ on $Q_{n}\setminus\{1\}$,  $a$ is a permutation of order $\ell =\lcm(m-1,n-1)$ on $S$.

Let $(p'_1,q_1)$ and $(p'_2,q_2)$ be elements of $S$, such that $p_1-q_1\equiv p_2-q_2 \mod g$. Then $p_2-p_1\equiv q_2-q_1\mod g$. 
Since $m-1$ and  $\frac{n-1}{g}$
are coprime,  by the Chinese Remainder Theorem the equivalences $k\equiv p_2-p_1\mod(m-1)$ and $ k\equiv q_2-q_1\mod \frac{n-1}{g}$ have a unique solution modulo $(m-1)\cdot\frac{n-1}{g}=\ell$ for $k$. 
Since $k\equiv p_2-p_1 \mod(m-1)$,  we have  $k\equiv p_2-p_1\equiv q_2-q_1\mod g$. Combined with $ k\equiv q_2-q_1\mod\frac{n-1}{g}$, this gives $k\equiv q_2-q_1\mod(n-1)$. Applying  $a^k$ to $(p'_1,q_1)$, gives the state $(p'_1+k\mod(m-1),q_1+k\mod(n-1))=(p'_2,q_2)$. 
So for every pair $(p'_1,q_1)$ and $(p'_2,q_2)$ such that $p_1-q_1\equiv p_2-q_2\mod g$,  $(p'_2,q_2)$  is reachable from  $(p'_1,q_1)$ by  some number of $a$'s.

If   $(p'_1,q_1)\in S$ is reachable,  all $(p',q)\in S$ such that $p_1-q_1\equiv p - q \mod g$  are also reachable. 
Since  $(p'_1,2)$ is reachable for $p_1=2,\dots,m-1$,  all $(p', q )\in S$ such that $p -q \in\{0,1,2,\dots,(m-3)\mod g\}$ are  reachable. 
Since $(2',3)$ is reachable, all $(p',q)\in S$ such that $p-q\equiv-1\equiv(m-2)\mod g$ are reachable. Since all remainders modulo $g$  have  at least one representative, all of $S$ is reachable.

We prove distinguishability using a number of claims:
\be
\item
{\bf $(2',2)$ is distinguishable from every other state in Column 2.}\hfill
	\be
	\item If the operation is intersection (symmetric difference), then $(m',n)$ is distinguishable from all other states in Column $n$, since $(m',n)$ is the only final (non-final) state in this column. Note that $a^{m-1}$ ($a^{n-1}$) acts as the identity on the set $Q'_m\setminus \{1'\}$  ($Q_n\setminus \{1\}$).
	Hence $a^{\lcm(m-1,n-1)}$ is the identity on 
	$Q'_m\setminus \{1\}\times Q_n\setminus \{1\}$, and
	$x=a^{\lcm(m-1,n-1)-1}$ maps $(2',2)$ to $(m',n)$.
	If two states are in the same column, then so are their successors after $a^\ell$ is applied, for any $\ell\ge 0$. Therefore applying $x$ to $(p',2)$ leads to a state in Column $n$; so $(2',2)$ and $(p',2)$ are distinguishable by $x$.
	
	\item
	If the operation is union (difference), then $(m',n-1)$ is distinguishable from any other state in Column $n-1$. 
	Consider $(2',2)$ and $(p',2)$; applying $ac$ results in $(3',2)$ and $(r', 2)$, where $r \neq 3$. Following this by $a^{\lcm(m-1,n-1)-2}$ yields $(m',n-1)$ and $(s', n-1)$, where $s \neq m$. Hence $aca^{\lcm(m-1,n-1)-2}$ distinguishes $(2',2)$ from $(p',2)$.
	\ee
\item
{\bf $(2',2)$ is distinguishable from every other state in Row 2.}\hfill\\
The argument is symmetric to that for Claim 1,  when the operation is intersection, symmetric difference, or union. If the operation is difference, the states can be distinguished by $a^{\lcm(m-1,n-1)-1}$ since this maps $(2',2)$ to $(m',n)$, which is distinguishable from all other states in Row $m$.
\item
{\bf For any two states in the same column there is a word mapping exactly one of them to $(2',2)$.}\hfill\\
Let the two states be $(p',q)$ and $(r',q)$. 
If $\{p, r\}=\{2, m\}$, let $s=1$; otherwise, let $s=0$. Applying $a^s$ 
results in $(p'_1,q_1)$ and $(r'_1,q_1)$, where $\{p_1,r_1\}\neq \{2,m\}$, since $m \ge 4$.
Thus we can assume that the pair of states to be distinguished is $(p',q)$ and $(r',q)$, where $\{p,r\}\neq \{2,m\}$. 
Now $c$ takes these states to $(p'_1 ,2)$ and $(r'_1,2)$, and $p_1 \neq r_1$, since $c$ can map two states of $Q'_m\setminus \{1'\}$ to the same state only if these states are $2'$ and $m'$. 
Observe that $ac$ cyclically permutes states $\{(p',2) \mid 2\le p\le m-1\}$.
So  applying $ac$  a sufficient number of times maps exactly one of the two states to $(2',2)$.
\item
{\bf For any two states in the same row there is a word mapping exactly one of them to $(2',2)$.}\hfill\\
The proof is symmetric to that for Claim 3, 
if we interchange rows and columns and replace $c$ by $e$.

\item
{\bf  For any pair of states, there exists a word that maps one of the states to $(2',2)$ and the other to a state $(p',2)$, $p\neq 2$, or $(2',q)$, $q \neq 2$.}
There are several cases:
\be
\item
If the states are in the same column or row, then the result holds by Claims 3 and 4.
Hence assume they  are not in the same column or row.
\item
If both states  
lie outside of Row $m$,
then $c$ takes both states to Column $2$ but it takes each state to a different row. The result now follows by Claim~3.

\item
If  both states  
lie outside of Column $n$,
then  $e$  takes both states to Row $2$ but it takes each state to a different column. The result now follows by Claim~4.

\item
If one of the states is $(m',n)$, then $a$ takes it to $(2',2)$. If Case (b) does not apply, the other state must have been taken to Row $m$.
If Case (c) also does not apply, it must also have been taken to Column $n$.
Thus the other state must have been taken to $(m',n)$. 
Applying $a$ again
results in $(2',2)$ and $(3',3)$, 
and this reduces to Case~(b).
\item
If the states are $(m',n-1)$ and $( (m-1)',n)$, then applying $a^2$ results in 
$(3',2)$ and $(2',3)$, 
and Case (b) applies.
\item
In the remaining case, one state is in Row $m$ and the other state is in Column $n$; furthermore, neither state is the state $(m',n)$ from Case (d), and the pair of states being considered is not the pair $(m',n-1)$ and $((m-1)',n)$ from Case (e).

If the state in Row $m$ is not $(m',n-1)$, applying $a$ sends it to a state that is not in Column $n$, and the other state to Column $2$; so Case (c) applies. 
Otherwise, the state in Row $m$ is $(m',n-1)$ and the state in Column $n$ is not $((m-1)',n)$. Applying $a$ sends the first state to Row 2 and the other state to a state not in Row $m$; so Case (b) applies.
\ee
\ee
We have shown that for any pair of states, there exists a word that takes one of the states to $(2',2)$ and the other state not to $(2', 2)$ but to either Row $2$ or Column $2$ by Claim 5.
By Claims 1 and 2, those states are distinguishable. Therefore the original states are also distinguishable.
\end{proof}

\begin{remark}
For regular left ideals, the minimal alphabet required to meet all the bounds has five letters. 
As was the case with right ideals, it is possible to reduce the alphabet for some operations~\cite{BJL13}. The sizes are as follows:
reversal (3/4), star (2/2), product (1/2), union $(4/3)$, intersection (2/3), symmetric difference (2/3), and difference (3/3). Note that the previously known witness for union used a four-letter alphabet, while ours only uses three letters. 
\end{remark}

\section{Two-Sided Ideals}
\label{sec:2sided}

The following stream of two-sided ideals was  defined in~\cite{BrYe11}:

\begin{definition}
\label{def:2sided}
For $n\ge 5$, let  
$\cD_n =\cD_n(a,b,c,d,e,f)= (Q_n,\Sig,\delta, 1,\{n\})$, where
$\Sig=\{a,b,c,d,e,f\}$, 
and $\delta_n$ is defined by the transformations
$a \colon (2,3,\ldots,n-1)$,
$b \colon (2,3)$,
$c \colon (n-1\to 2)$,
$d \colon (n-1\to 1)$,
$e \colon (Q_{n-1}\to 2)$,
and $f \colon (2\to n)$.
The structure of  $\cD_n(a,b,c,d,e,f)$ is shown in Figure~\ref{fig:2sided}. 
\end{definition}

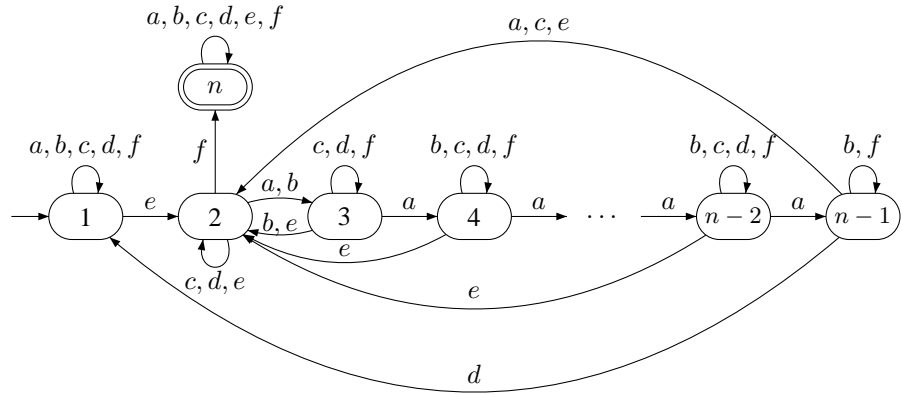
\begin{figure}[ht]
\unitlength 7pt
\begin{center}\begin{picture}(43,19)(0,-1)
\gasset{Nh=2.5,Nw=4,Nmr=1.25,ELdist=0.4,loopdiam=1.5}
\node(n)(8,14){$n$}\rmark(n)
\drawloop(n){$a,b,c,d,e,f$}

\node(1)(1,7){1}\imark(1)
\node(2)(8,7){2}
\node(3)(15,7){3}
\node(4)(22,7){4}
\node[Nframe=n](4dots)(29,7){$\dots$}
	{\small
\node(n-2)(36,7){$n-2$}
	}
	{\small
\node(n-1)(43,7){$n-1$}
	}
\drawedge(2,n){$f$}
\drawedge(1,2){$e$}
\drawloop(1){$a,b,c,d,f$}
\drawloop[loopangle=270,ELdist=.2](2){$c,d,e$}
\drawedge[curvedepth= 1,ELdist=.1](2,3){$a,b$}
\drawedge[curvedepth= 1,ELdist=-1.2](3,2){$b,e$}
\drawloop(3){$c,d,f$}
\drawedge(3,4){$a$}
\drawedge[curvedepth= 2.5,ELdist=-1](4,2){$e$}
\drawedge(4,4dots){$a$}
\drawedge(4dots,n-2){$a$}
\drawloop(4){$b,c,d,f$}
\drawloop(n-2){$b,c,d,f$}
\drawedge(n-2,n-1){$a$}
\drawedge[curvedepth= 5,ELdist=-1.2](n-2,2){$e$}
\drawedge[curvedepth= -9.5,ELdist=-1.2](n-1,2){$a,c,e$}
\drawedge[curvedepth= 9.5,ELdist=-1.5](n-1,1){$d$}
\drawloop(n-1){$b,f$}
\end{picture}\end{center}
\caption{Minimal DFA $\cD_n(a,b,c,d,e,f)$  of Definition~\ref{def:2sided}.}
\label{fig:2sided}
\end{figure}

This stream of languages is closely related to the stream of Figure~\ref{fig:RegWit}.
The DFA $\cD_n(a,b,c,d,e,f)$ of Definition~\ref{def:2sided} is constructed by taking $\cD_{n-2}(a,b,c)$ of Figure~\ref{fig:RegWit} with states relabeled $2,\dots,n-1$, adding new states $1$ and $n$, and new inputs $d\colon (n-1 \rightarrow 1)$,  $e\colon (Q_{n-1} \to 2)$, and $f \colon (2\to n)$.

\begin{theorem}[Two-Sided Ideals]
For each $n\ge 5$,  
the DFA $\cD_n(a,b,c,d,e,f)$ of Definition~\ref{def:2sided} is minimal and its 
language $L_n(a,b,c,d,e,f)$ is a two-sided ideal of complexity $n$.
The stream $(L_n(a,b,c,d,e,f) \mid n \ge 5)$  with dialect stream
$(L_n(b,a,c,d,e,f) \mid n \ge 5)$
is most complex in the class of regular two-sided ideals.
In particular, this stream meets all the complexity bounds listed below, which are maximal for two-sided ideals. In several cases the bounds can be met with restricted alphabets, as shown below.
\be
\item
The syntactic semigroup of $L_n(a,b,c,d,e,f)$ has cardinality $n^{n-2}+(n-2)2^{n-2}+1$.  Moreover, fewer than six inputs do not suffice to meet this bound.
\item
All quotients of $L_n(a,-,-,d,e,f)$ have complexity $n$, except the quotient corresponding to state $n$, which has complexity $1$.
\item
$L_n(a,-,-,d,e,f)$ has $2^{n-2}+1$ atoms.

\item
For each atom $A_S$ of $L_n(a,b,c,d,e,f)$, the complexity $\kappa(A_S)$ satisfies:
\begin{equation*}
	\kappa(A_S) =
	\begin{cases}
		 n, 			& \text{if $S=Q_n$;}\\
		 2^{n-2}+n-1,		& \text{if $S=Q_n\setminus \{1\} $;}\\
		  1 + \sum_{x=1}^{|S|}\sum_{y=1}^{n-|S|}\binom{n-2}{x-1}\binom{n-x-1}{y-1},
		 			& \text{otherwise.}
		\end{cases}
\end{equation*}
\item
The reverse of $L_n(a,-,-,d,e,f)$ has complexity $2^{n-2}+1$.

\item
The star of $L_n(a,-,-,-,e,f)$ has complexity $n+1$.
\item
The product $L_m(a,-,-,-,e,f) L_n(a,-,-,-,e,f)$ has complexity $m+n-1$.
\item
For any proper binary boolean function $\circ$, the complexity of $L_m(a,b,-,d,e,f) \\ \circ L_n(b,a,-,d,e,f)$
is maximal. In particular,
	\be
	\item
	 $L_m(a,b,-,d,e,f) \cap L_n(b,a,-,d,e,f)$ has complexity $mn$, as does\\ 
	 $L_m(a,b,-,d,e,f) \oplus L_n(b,a,-,d,e,f)$.
	 \item
	$L_m(a,b,-,d,e,f) \setminus L_n(b,a,-,d,e,f)$ has complexity 	$mn
	- (m-1)$.
	\item
	$L_m(a,b,-,d,e,f) \cup L_n(b,a,-,d,e,f)$ has complexity 	$mn
	- (m+n-2)$.
	\item 
	If $m\neq n$, the bounds are met by $L_m(a,b,-,d,e,f)$ and $ L_n(a,b,-,d,e,f)$.
	\ee
\ee
\end{theorem}
\begin{proof}
Notice that inputs $a$, $e$ and $f$ are needed to make all the states reachable.
It was proved in~\cite{BrYe11} that $\cD_n(a,b,c,d,e,f)$  is minimal and accepts a two-sided ideal.
\be
\item
It was shown in~\cite{BrSz14} that the syntactic semigroup of a two-sided ideal of complexity $n$ has cardinality at most $n^{n-2}+(n-2)2^{n-2}+1$, and
in~\cite{BrYe11} that the syntactic semigroup of $L_n(a,b,c,d,e,f)$ has that cardinality. 
Moreover, it was proved in~\cite{BSY15} that an alphabet of at least six letters is required to meet this bound. 
\item
This follows from Definition~\ref{def:2sided}.
\item
The number of atoms of any regular language $L$ is equal to the complexity of the reverse of $L$~\cite{BrTa13}. It was proved in~\cite{BrYe11} that the complexity of the reverse of $L_n(a,-,-,d,e,f)$ is $2^{n-2}+1$.
\item
This was proved in~\cite{BrDa15}.

\item
See Item~3 above.
\item
The argument is the same as for the star of right ideals.
\item
See Theorem~\ref{thm:2Prod}.
\item
See Theorems~\ref{thm:2Bool} and~\ref{thm:2Bool2}.
\ee
\end{proof}

\subsection{Product}
\label{ssec:product_2sided}
We show that the complexity of the product of the DFA $\cD'_m(a,-,-,-,e,f)$ with $\cD_n(a,-,-,-,e,f)$ 
meets the upper bound $m+n-1$ derived in~\cite{BJL13}. We use the same construction as for left ideals for the product DFA $\cD$. 
The construction is illustrated in Figure \ref{fig:TSProd} for $m=n=5$.

\begin{theorem} [Two-Sided Ideals: Product]
\label{thm:2Prod}
For $m,n\ge 5$, the product of the language $L_m(a,-,-,-,e,f)$ with $L_n(a,-,-,-,e,f)$ has complexity $m+n-1$.
\end{theorem}
\begin{proof}
By construction $\cD$ has $m+n-1$ states.
We know that all the states in $\cD_m'$ are pairwise distinguishable. Hence in $\cD$, for each pair there exists a word that takes one state to state 1 and the other to a state in $Q'_m\setminus\{m'\}$. Also, all pairs of distinct states in $\cD_n$ are distinguishable. 
To distinguish a state $p'$ from a state $q$ in $\cD$, note that every word accepted from $p'$ contains two $f$'s, whereas there are words accepted from $q$ that contain only one $f$.

\end{proof}

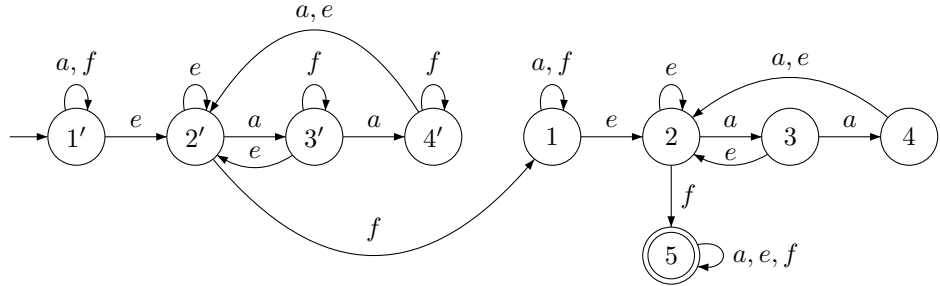
\begin{figure}[ht]
\unitlength 9pt
\begin{center}\begin{picture}(36,10)(1,-1)
\gasset{Nh=2.4,Nw=2.4,Nmr=1.2,ELdist=0.4,loopdiam=1.0}
\node(1')(2,4){$1'$}\imark(1')
\node(2')(7,4){$2'$}
\node(3')(12,4){$3'$}
\node(4')(17,4){$4'$}

\drawedge(1',2'){$e$}
\drawedge(2',3'){$a$}
\drawedge(3',4'){$a$}

\drawedge[curvedepth=1.3,ELdist=-.9](3',2'){$e$}
\drawedge[curvedepth=-4.5,ELdist=-1.0](4',2'){$a,e$}

\drawloop(1'){$a,f$}
\drawloop(2'){$e$}
\drawloop(3'){$f$} 
\drawloop(4'){$f$} 

\node(1)(22,4){$1$}
\node(2)(27,4){$2$}
\node(3)(32,4){$3$}
\node(4)(37,4){$4$}
\node(5)(27,-1){$5$}\rmark(5)

\drawedge(1,2){$e$}
\drawedge(2,3){$a$}
\drawedge(3,4){$a$}
\drawedge(2,5){$f$}

\drawedge[curvedepth= -5,ELdist=.6](2',1){$f$}

\drawedge[curvedepth= 1.2,ELdist=-.7](3,2){$e$}
\drawedge[curvedepth=-2.5,ELdist=-1.0](4,2){$a,e$}

\drawloop(1){$a,f$}
\drawloop(2){$e$}
\drawloop[loopangle=360](5){$a,e,f$}
\end{picture}\end{center}
\caption{Product of two-sided-ideal witnesses with $m=5$, $n=5$.}
\label{fig:TSProd}
\end{figure}

\subsection{Boolean Operations}
\label{ssec:boolean_2sided}
\begin{theorem}[Two-Sided Ideals: Boolean Operations]
\label{thm:2Bool}
If $m,n\ge 5$, then
\be
\item
The complexity of $L_m(a,b,-,d,e,f) \cap L_n(b,a,-,d,e,f)$ is $mn$.
\item
The complexity of $L_m(a,b,-,d,e,f) \oplus L_n(b,a,-,d,e,f)$ is $mn$.
\item
The complexity of $L_m(a,b,-,d,e,f) \setminus L_n(b,a,-,d,e,f)$ is $mn-(m-1)$.
\item
The complexity of $L_m(a,b,-,d,e,f) \cup L_n(b,a,-,d,e,f)$ is $mn-(m+n-2)$.
\ee
\end{theorem}
\begin{proof}
As before, we take the direct product DFA $\cD_{m,n} = \cD'_m(a,b,-,d,e,f) \times \cD_n(b,a,-,d,e,f)$ and count reachable and distinguishable states.
First we prove all states are reachable, as illustrated in Figure~\ref{fig:2Cross}. State $(1',1)$ is the initial state, and $(2', 2)$ is reached from $(1', 1)$ by  $e$. Since $\{a,b\}$ generates all permutations of $Q'_m\setminus\{1',m'\}$ and $Q_n\setminus\{1,n\}$, by~\cite[Theorem 1]{BBMR14} all states in $(Q'_m\setminus\{1',m' \})\times(Q_n\setminus\{1,n\})$ are reachable, unless $(m-2,n-2) \in \{(2,2),(3,4),(4,3),(4,4)\}$; the case $(2,2)$ does not occur since $m,n \ge 5$, and we have verified the other special cases computationally. Thus it remains to show that all states in Rows $1$ and $m$ and Columns $1$ and $n$ are reachable.

For Row 1, first note that $(1',2)$ is reachable since $( (m-1)', 2)d = (1',2)$.
Then $( 1', q )b = (1', q+1)$
for $2 \le q \le n-2$, so we can reach $(1',3),\dotsc,(1',n-1)$.
Finally, we can reach $(1',n)$ since $( 1', 2)f =  (1',n) $.
A symmetric argument applies to Column 1.
For Row $m$, note $(m',1)$ is reachable since it is in Column 1.
Then we have $(m',1)e = (m',2)$, and $( m', q )b = (m',q+1) $, for $2 \le q \le n-2$, and finally $(m',2)f = (m',n)$.
A symmetric argument applies to Column $n$.

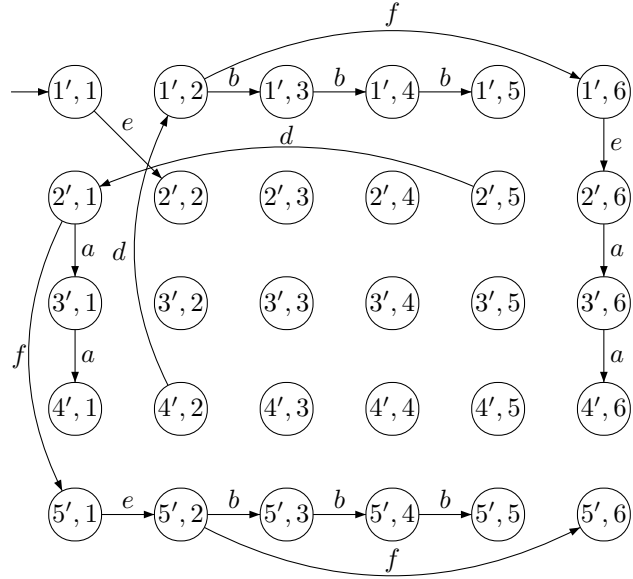
\begin{figure}[ht]
\unitlength 8pt
\begin{center}\begin{picture}(30,27)(0,-8)
\gasset{Nh=2.5,Nw=2.5,Nmr=1.2,ELdist=0.3,loopdiam=1.2}
\node(1'1)(2,15){$1',1$}\imark(1'1)
\node(2'1)(2,10){$2',1$}
\node(3'1)(2,5){$3',1$}
\node(4'1)(2,0){$4',1$}
\node(5'1)(2,-5){$5',1$}

\node(1'2)(7,15){$1',2$}
\node(2'2)(7,10){$2',2$}
\node(3'2)(7,5){$3',2$}
\node(4'2)(7,0){$4',2$}
\node(5'2)(7,-5){$5',2$}

\node(1'3)(12,15){$1',3$}
\node(2'3)(12,10){$2',3$}
\node(3'3)(12,5){$3',3$}
\node(4'3)(12,0){$4',3$}
\node(5'3)(12,-5){$5',3$}

\node(1'4)(17,15){$1',4$}
\node(2'4)(17,10){$2',4$}
\node(3'4)(17,5){$3',4$}
\node(4'4)(17,0){$4',4$}
\node(5'4)(17,-5){$5',4$}

\node(1'5)(22,15){$1',5$}
\node(2'5)(22,10){$2',5$}
\node(3'5)(22,5){$3',5$}
\node(4'5)(22,0){$4',5$}
\node(5'5)(22,-5){$5',5$}

\node(1'6)(27,15){$1',6$}
\node(2'6)(27,10){$2',6$}
\node(3'6)(27,5){$3',6$}
\node(4'6)(27,0){$4',6$}
\node(5'6)(27,-5){$5',6$}

\drawedge[ELpos=40](1'1,2'2){$e$}
\drawedge(1'2,1'3){$b$}
\drawedge(1'3,1'4){$b$}
\drawedge(1'4,1'5){$b$}
\drawedge(2'1,3'1){$a$}
\drawedge(3'1,4'1){$a$}
\drawedge(1'6,2'6){$e$}
\drawedge(5'1,5'2){$e$}
\drawedge[curvedepth=2.2,ELdist=0.4](4'2,1'2){$d$}
\drawedge[curvedepth=-2.4,ELdist=-1](2'5,2'1){$d$}
\drawedge[curvedepth=2.9,ELdist=0.2](1'2,1'6){$f$}
\drawedge[curvedepth=-2.2,ELdist=-0.8](2'1,5'1){$f$}

\drawedge[curvedepth= - 2.9,ELdist=0.2](5'2,5'6){$f$}
\drawedge(5'2,5'3){$b$}
\drawedge(5'3,5'4){$b$}
\drawedge(5'4,5'5){$b$}
\drawedge(2'6,3'6){$a$}
\drawedge(3'6,4'6){$a$}
\end{picture}\end{center}
\caption{Partial direct product for boolean operations on two-sided ideals. 
}
\label{fig:2Cross}
\end{figure}

Before we begin the distinguishability proofs, we make a few observations. Let $(p',q)$ and $(r',s)$ be states with $1\le p < r \le m$; note that $r > 1$.
\be
\item
If $r < m$, the word $a^{m-r}$ sends $(r',s)$ to a state in Row $2$. It sends $(p',q)$ to either Row 1 (if $p = 1$) or Row $p+m-r$ (if $p \ge 2$); in either case $(p',q)$ is not sent to Row $2$ or Row $m$.
\item
If $r < m$, the word $a^{m-r}f$ sends $(r',s)$ to Row $m$, and it  sends $(p',q)$ to the same  row as $a^{m-r}$.
\item
If $r < m$, then $a^{m-r}f^{m-r} = a^{m-r}f$ sends $(r',s)$ to Row $m$, and $(p',q)$ to a row other than Row $m$.  If $r = m$, then $a^{m-r}f^{m-r} = \eps$, the state $(r',s) = (m',s)$ is in Row $m$, and the state $(p',q)$ is not in Row $m$. 
\ee
Thus the word $a^{m-r}f^{m-r}$ will send any state in Row $r$ to Row $m$, and any state not in Row $r$ to a row other than Row $m$. By a symmetric argument, the word $b^{n-s}f^{n-s}$ sends states in Column $s$ to Column $n$, and states not in Column $s$ to a different column, for $1 \le q < s \le n$. We use this fact frequently to distinguish states.

For each boolean operation we now prove that the number of distinguishable states meets the relevant upper bound.
\be
\item
{\bf Intersection/Symmetric Difference.}
For intersection, the only final state is $(m',n)$. For symmetric difference, every state in Row $m$ or Column $n$ except  $(m',n)$ is final, and $(m',n)$ is the only empty state. We can use the same distinguishing words in both cases.
	\be
	\item
	{\bf States in distinct rows, $p < r \le m$}: States $(p', q)$ and $(r', s)$, $q, s \in Q_n $, are distinguishable by first applying $a^{m-r}f^{m-r}$. This sends one of the states to Row $m$, and the other to a different row.  We then apply $ea^2f$: applying $e$ sends the state in Row $m$ to $(m',2)$  or $(m',n)$, and the state not in Row $m$ to $(2',2)$ or $(2', n)$. Then $a^2$ fixes $(m',2)$  or $(m',n)$ and sends the other state to  $(4',2)$ or $(4',n)$. Finally, applying $f$ sends  one state to $(m',n)$, and the other state to $(4',n)$. 
	\item
	{\bf States in distinct columns, $q < s \le n$}: The argument is symmetric if we interchange $a$ and $b$.
	\ee
Since all states are distinguishable, the complexity  is $mn$.

\item
{\bf Difference.}
Here the final states are those which are in Row $m$, but not Column $n$.
States in Column $n$ are indistinguishable and empty, since no word can take any of these states out of Column $n$, and that column 	contains only non-final states. Hence there are at most $mn-(m-1)$ distinguishable states.
	\be
	\item
	{\bf States in Row $m$, $1\le q < s  \le n$}: States $(m',q)$ and $(m',s)$ are distinguishable by $b^{n-s}f^{n-s}$, since $(m',n)$ is the only non-final state in Row $m$.
	\item
	{\bf States in distinct rows, and one state is in Column $n$}:
Any state outside of Column $n$ accepts $eb^2f$, and thus all of these states are non-empty. It follows that all states outside of Column $n$ are distinguishable from the empty states in Column $n$. 
	\item
	{\bf States in distinct rows, and neither state is in Column $n$}: 
    States $(p',q)$ and $(r',s)$, with $p < r \le m$ and $q,s \le n-1$, are distinguishable by $a^{m-r}f^{m-r}$ unless this word sends  $(r',s)$ to $(m',n)$. This occurs only if $r < m$ and $a^{m-r}$ sends  $(r',s)$ to $(2',2)$. In this case, applying $b^2$ after $a^{m-r}$  sends  $(2',2)$ to $(2',4)$, and does not affect the row numbers; thus one of them will be in Row 2 and the other in neither Row 2 nor Row $m$. Then $f$ distinguishes the states.
	\item
	{\bf States in distinct columns}: The arrangement of final and non-final states in Row $m$ matches that of symmetric difference. Thus the argument used for intersection/symmetric difference  applies here also.
	\ee
Since all states in $Q'_m\times(Q_n\setminus\{n\})$ are distinguishable, the complexity of difference is $mn-(m-1)$.

\item
{\bf Union.}
The final states are those in Row $m$ and Column $n$.
All final states are indistinguishable, since they all accept $\Sig^*$.
Therefore there are at most $1+(m-1)(n-1)=mn-(m+n-2)$ distinguishable states. 
\be
\item
{\bf Non-final states in distinct rows}:
Two non-final states $(p',q)$ and $(r',s)$,  with $p < r < m$  and $q,s < n$, are distinguishable by $a^{m-r}f$, unless this word sends $(p',q)$ to Column $n$ (since it also sends $(r',s)$ to Row $m$). This occurs only if $a^{m-r}$ sends $(p',q)$ to Column 2. In this case we can use $a^{m-r}b^2f$ to distinguish the states; this is similar to Case (c) of the difference operation.
\item
{\bf Non-final states in distinct columns}:
These states are distinguishable by a symmetric argument.
\ee
Thus the complexity of union is $mn-(m+n-2)$. 
\ee
\end{proof}

If $m\neq n$, the complexity bounds for boolean operations can be met by using languages from the stream $(L_n(a,b,-,d,e,f) \mid n \ge 5)$ for both arguments. That is, we can meet the bounds for boolean operations without the dialect stream $(L_n(b,a,-,d,e,f) \mid n \ge 5)$.

\goodbreak

\begin{theorem}[Two-Sided Ideals: Boolean Operations, $m\neq n$]
\label{thm:2Bool2}
Suppose $m,n\ge 5$ and $m \ne n$.
\be
\item
The complexity of $L_m(a,b,-,d,e,f) \cap L_n(a,b,-,d,e,f)$ is $mn$.
\item
The complexity of $L_m(a,b,-,d,e,f) \oplus L_n(a,b,-,d,e,f)$ is $mn$.
\item
The complexity of $L_m(a,b,-,d,e,f) \setminus L_n(a,b,-,d,e,f)$ is $mn-(m-1)$.
\item
The complexity of $L_m(a,b,-,d,e,f) \cup L_n(a,b,-,d,e,f)$ is $mn-(m+n-2)$.
\ee
\end{theorem}

\begin{proof}
Here we consider the DFA $\cD_{m,n} = \cD'_m(a,b,-,d,e,f) \times \cD_n(a,b,-,d,e,f)$.
The proof that all states are reachable is identical to the proof above, except state $(1,n-1)$ is reachable from $(1,n-2)$ by $a$ instead of $b$, and when applying~\cite[Theorem~1]{BBMR14} the special cases we must verify are only $(m-2,n-2) \in \{(3,4),(4,3)\}$, since we are assuming $m \ne n$.
Also, for states $(p',q)$ and $(r',s)$ with $p < r$, the same remark about the word $a^{m-r}f^{m-r}$ applies, i.e., this word sends $(r',s)$ to Row $m$ and $(p',q)$ to a different row. For $(p',q)$ and $(r',s)$ with $q < s$, the word $a^{n-s}f^{n-s}$ sends $(r',s)$ to Column $n$ and $(p',q)$ to a different column (previously we used $b^{n-s}f^{n-s}$ for this purpose).

For each boolean operation we now prove that the number of distinguishable states meets the relevant upper bound.
\be
\item
{\bf Intersection/Symmetric Difference.}
As before, the same distinguishing words can be used for intersection and symmetric difference.
\be
\item
{\bf States in Column $n$}: States $(p', n)$ and $(r', n)$, with $p < r \le m$, are distinguishable by $a^{m-r}f^{m-r}$.
\item
States $(2',2)$ and $(m',2)$ are distinguishable as follows. Since all states in 
$(Q'_m \setminus \{1',m'\}) \times (Q_n \setminus \{1,n\})$
 are reachable from $(2',2)$ using words in $\{a,b\}^*$, there is a word $w \in \{a,b\}^*$ which sends $(2',2)$ to $(3',2)$. 
If we view $w$ as a permutation on $Q_n \setminus \{1,n\}$ alone, it must fix $2$.
Thus $w$ sends $(m',2)$ to $(m',2)$. 
Then $(3',2)$ and $(m',2)$ are distinguished by $f$.
\item
{\bf States in Column $q<n$}: States $(p',q)$ and $(r',q)$, with $p<r\le m$, are distinguishable since the word $a^{m-r}f^{m-r}e$ reduces this case to Case (a) or (b).
\item
By symmetry, states in the same row are distinguishable.
\item
{\bf States in distinct columns}: $(p', q)$ and $(r', s)$, $q < s \le n$, are distinguishable  by first applying $a^{n-s}f^{n-s}e$. This sends $(p',q)$ to some state in Column 2 and $(r',s)$ to some state in Column $n$. If these successor states  are in the same row, then this case reduces to Case (d). Otherwise, since $e$ was applied, the only possible states are $(2',2)$ and $(m',n)$, or $(2',n)$ and $(m',2)$. In either case, these states are distinguished by $a^{\min(m,n)-2}f$. 
\item
By symmetry, states in distinct rows are distinguishable.
\ee
Since all states are distinguishable, the complexity of intersection is $mn$.

\item
{\bf Difference.}
States in Column $n$ are empty and thus indistinguishable.
\be
\item
{\bf States in distinct columns, and one state is in Column $n$}:
All states $(p',q)$, $q < n$ are non-empty, and thus distinguishable from those in Column $n$. To see this, observe that if $p = m$, then $(p',q)$ is a final state and thus is non-empty. If $p < m$ then $e$ sends $(p',q)$ to $(2',2)$. Since every state in 
$(Q'_m \setminus \{1',m'\}) \times (Q_n \setminus \{1,n\})$ is reachable from $(2',2)$, there is a word $w \in \{a,b\}^*$ that sends $(2',2)$ to $(2',3)$. Then $f$ sends $(2',3)$ to the final state $(m',3)$.
\item
{\bf States in distinct columns, and neither state is in Column $n$}:
States $(p',q)$ and $(r',s)$, with $q < s < n$ are distinguishable since applying $a^{n-s}f$ reduces this case to 
Case~(a).
\item
{\bf States in distinct rows, and neither state is in Column $n$}:
States $(p',q)$ and $(r',s)$, with $p < r \le m$ and $q,s$ not both equal to $n$, are distinguishable by $a^{m-r}f^{m-r}$, unless this word sends  $(r',s)$ to $(m',n)$.
This occurs only if  $r < m$ and $a^{m-r}$ sends $(r',s)$ to $(2',2)$.
Recall that $a^{m-r}$ sends $(r',s)$ to Row 2, and sends $(p',q)$ to a row other than Row 2 or Row $m$; so suppose that after applying $a^{m-r}$ the states are $(2',2)$ and  $(i',k)$ for some $i \not\in \{2,m\}$.
There is a word $w \in \{a,b\}^*$ which sends $(2',2)$ to $(2',3)$. 
If we view $w$ as a permutation on $Q'_m \setminus \{1',m'\}$, it sends $2'$ to $2'$, so it sends $i'$ to some $j' \not\in \{2',m'\}$.
So $w$ sends $(2',2)$ and $(i',k)$ to $(2',3)$ and $(j',\ell)$ respectively; since $j' \not\in \{2',m'\}$ these states are distinguished by $f$.
\ee
Since $(Q'_m\times(Q_n\setminus\{n\}))\cup\{(m',n)\}$ is a maximal distinguishable set, the complexity of set difference is $m(n-1)+1=mn-(m-1)$.

\item
{\bf Union.}
All final states are indistinguishable, since they all accept $\Sig^*$.
\be
\item
{\bf Non-final states in the same row}:
Consider states $(p',q)$ and $(p',s)$ with $p < m$ and $q < s < n$.
The word $a^{n-s}$ takes these to $(r',q+n-s)$  or $(r',1)$, and $(r',2)$ respectively for some $r$; note  that in either case $(p',q)$ does not get mapped to Column $2$.
Hence the two states are sent to $(r',i)$ where $i \ne 2$ and $(r',2)$, respectively.
If $r \ne 2$, then these states are distinguished by $f$.

If $r = 2$, there exists a word $w \in \{a,b\}^*$ which maps $(r',2) = (2',2)$ to $(3',2)$.
We can view $w$ as a permutation on $Q_n$ which fixes $1$ and $n$; it sends $2$ to $2$ and sends $i$ to some $j \ne 2$, since it is bijective. 
Hence the other state $(r',i)$ is sent to $(3',j)$ for some $j \ne 2$.
Then $(3',2)$ and $(3',j)$ are distinguished by $f$.
\item
{\bf Non-final states in the same column}: These are distinguishable by a symmetric argument.

\item
{\bf Non-final states in distinct columns}:
Non-final states $(p',q)$ and $(r',s)$, with  $p , r < m$ and $q < s < n$, are distinguishable by $a^{n-s}f$, unless this word sends $(p',q)$ to Row $m$ (since it also sends $(r',s)$ to Column $n$). This occurs only if $a^{n-s}$ sends $(p',q)$ to Row $2$. 
Note that $a^{n-s}$ sends $(r',s)$ to Column $2$; so after applying $a^{n-s}$ the states are $(i',2)$ and $(2',j)$ for some $i$ and $j$.  By (a) and (b), we can assume $i,j \neq 2$. Then $a^{\min(m,n)-2}f$ distinguishes the states.
\item
{\bf Non-final states in distinct rows}:
A symmetric argument works.
\ee
Since all non-final states are distinguishable and all final states are indistinguishable, the complexity of union is $(m-1)(n-1)+1=mn-(m+n-1)$. 

\ee
\end{proof}

\begin{remark}
For regular two-sided ideals, the minimal alphabet required to meet all the bounds has six letters. 
As before, it is possible to reduce the alphabet for some operations~\cite{BJL13}. The sizes are as follows:
reversal (3/4), star (2/3), product (1/3), union (2/5), intersection (2/5), symmetric difference (2/5), and difference (2/5). 

\end{remark}

\section{Conclusions}
In the case of regular right, left, and two-sided ideals,
we have demonstrated that there exist  single witness streams that meet the bounds for all of our complexity measures, and that the only dialects required are those in which two  letters are interchanged; this is needed in the case where the bounds for boolean operations are to be met with witnesses of the same complexity.
\providecommand{\noopsort}[1]{}

\end{document}